\documentclass[letter,onecolumn,12pt]{IEEEtran}
\usepackage{cite,graphicx,amsmath,amssymb}
\usepackage{subfigure}
\usepackage{citesort}
\usepackage{fancyhdr}
\usepackage{tabularx}
\usepackage{color}

\newtheorem{theorem}{Theorem}

\newtheorem{lemma}{Lemma}

\newtheorem{corollary}{Corollary}

\newtheorem{proposition}{Proposition}

\newtheorem{conjecture}{Conjecture}

\newtheorem{sketch}{Sketch of Proof}

\newcommand{\defeq}{\stackrel{\Delta}{=}}

\newcommand{\papertitle}{Maximum Sum-Rate of MIMO Multiuser Scheduling with Linear Receivers}

\begin{document}

%%%%%%%%%%%%%
%\IEEEaftertitletext{\vspace{-20pt}}
% TITLE
\title{\papertitle}
\author{Raymond H.\ Y.\ Louie${}^{\dagger *}$,  Matthew R.\ McKay${}^{\ddagger}$, Iain B.\ Collings$^{*}$
\\
{\small
${}^\dagger$Telecommunications Lab, School of Electrical and Information Engineering, University of Sydney, Australia \\
${}^*$Wireless Technologies Laboratory, ICT Centre, CSIRO, Sydney, Australia \\
${}^\ddagger$Department of Electronic and Computer Engineering, \\ Hong Kong University of Science and Technology, Hong Kong \\
}} \maketitle

 \setcounter{page}{1} \thispagestyle{empty}
%\vspace*{-1.5cm}

 %%%%%%%%%%%%
% ABSTRACT
\begin{abstract}
We analyze scheduling algorithms for multiuser communication systems with users having multiple antennas and linear receivers.  When there is no feedback of channel information, we consider a common round robin scheduling algorithm, and derive new exact and high signal-to-noise ratio (SNR) maximum sum-rate results for the maximum ratio combining (MRC) and minimum mean squared error (MMSE) receivers. We also present new analysis of MRC, zero forcing (ZF) and MMSE receivers in the low SNR regime. When there are limited feedback capabilities in the system, we consider a common practical scheduling scheme based on signal-to-interference-and-noise ratio (SINR) feedback at the transmitter. We derive new accurate approximations for the maximum sum-rate, for the cases of MRC, ZF and MMSE receivers.
We also derive maximum sum-rate scaling laws, which reveal
that the maximum sum-rate of all three linear receivers converge to the same value for a large number of users, but at different rates.
\end{abstract}

\vspace*{0.2cm}

\noindent \noindent {\bf Corresponding Author:} \\
\noindent Raymond H.\ Y.\ Louie\\
Sch.\ of Elec.\ and Info. Engineering, University of Sydney, NSW 2006, Australia\\
+61-2-9372 4244
 (phone)    +61-2-9351 3847
 (fax), {\it rlouie@ee.usyd.edu.au}

 {\small The material in this paper has been presented in part at the IEEE Global Communications Conference (GLOBECOM), New Orleans, LA, USA, December 2008.}

\newpage
\setcounter{page}{1}

\section{Introduction}

%The performance of wireless communication systems can be
%significantly improved by employing antenna arrays at the receiver
%in conjunction with adaptive linear diversity combining strategies.
%Linear strategies are of interest because of their low complexity
%and ease of implementation.  In multiple-input multiple-output
%(MIMO) systems, where the transmitter also has multiple antennas,
%the performance of the system can be significantly improved by
%allowing for a higher data rate or spatial diversity.

The use of multiple-input multiple-output (MIMO) antenna systems in
multiuser environments has recently gained considerable attention.
Much progress has been made on characterizing the fundamental limits
and developing signal processing schemes for both the
multiple-access channel and broadcast channel scenarios (see, eg.\ \cite{choi04,jindal05}). In this paper, we consider the broadcast channel, which embraces
a large number of important application scenarios; for example, a
base-station serving multiple users in a cellular environment, or a
relay node serving multiple network nodes in an ad-hoc network.

For systems with perfect feedback capabilities, i.e.\ when complete short-term channel state
information (CSI) of all of the users are available to the
transmitter, it is well-known that the optimum solution for
achieving the capacity is to employ dirty paper coding (DPC).
However, while DPC is capacity achieving, it also has prohibitive
complexity. This, coupled with the stringent and potentially
overwhelming feedback requirements, render DPC unsuitable for
practical implementation.  Various low complexity alternatives have
been employed and studied extensively (see, eg.\ \cite{wong03,choi04,yoo06}), however these schemes all require perfect knowledge of the CSI at the transmitter.  For increasing number of users and antennas, the amount of CSI feedback becomes large, and these schemes become impractical. In addition, inaccurate CSI feedback may lead to significant performance losses \cite{wang06}.

%In addition, inaccurate CSI may lead to significant performance loss, as shown in \cite{wang07} where the error performance hits a noise floor at high signal-to-noise ratio (SNR) when zero-forcing (ZF) receivers are used in MIMO systems.

% Various low complexity alternatives have
%been employed and studied extensively, including transmit ZF
%beamforming \cite{yoo06} and block-diagonalization precoding \cite{choi04}.
%However, each of these schemes still require full CSI of all the users
%at the transmitter, and are not applicable when there are more users than the number of transmit antennas.

%In this paper, we are interested in situations where there are more users than the number of transmit antennas.

In this paper, we are interested in the practical scenarios when there is limited or no feedback of CSI at the transmitter. In this case, efficient scheduling algorithms are
required for coordinating transmissions to the different users.
Appropriate scheduling approaches depend on the feedback
capabilities between the receivers and the transmitter, and their performance can be measured in terms of the maximum sum-rate.

 %We are also interested in the practical case when there is limited or no feedback of CSI to the transmitter.

For systems with no feedback capabilities,
an obvious and practical scheduling approach is to
employ a simple round-robin scheduling algorithm; i.e.\ a
time-division multiple access (TDMA) scheme. The maximum sum-rate in
this case becomes equivalent to the maximum sum-rate of a single user MIMO
system. The single user MIMO maximum sum-rate has been extensively studied for optimal receivers
(see, eg.\ \cite{foschini98,mckay_jnl05}), and zero forcing (ZF) linear receivers (see, eg.\ \cite{chen07,forenza07}). Other results include the particular case where the number of transmit antennas is less than or equal to the number of receive antennas in the high signal-to-noise ratio (SNR) regime, for which the maximum sum-rate has been derived for ZF and minimum mean squared error (MMSE) receivers in \cite[Eq. 8.54]{tse05}. In this paper, we derive new exact maximum sum-rate expressions for both the maximum ratio combining (MRC) and MMSE receivers.  We also derive new maximum sum-rate expressions in the high SNR regime for MMSE receivers in the case where the number of transmit antennas is greater than the number of receive antennas, and for MRC receivers in the case of arbitrary antenna configurations. We then investigate the maximum sum-rate in the low SNR regime for each of the three linear receivers.

When limited feedback is available, the scheduling algorithm can be more sophisticated. In this paper we consider the common case of systems which feed back signal-to-interference-and-noise ratios (SINRs) and use opportunistic scheduling.  In particular, we examine systems where the SINR (for each user) can be measured from each transmit antenna separately, and then fed back.  In the scheduling algorithm each transmit antenna sends an independent data stream to the user with the largest corresponding SINR.  Of particular interest is systems employing practical linear receivers.  For
the specific case of ZF, previous related results have been
presented in \cite{airy04}, where an upper bound to the asymptotic
maximum sum-rate scaling law was derived, and in \cite{chen07}, where an
exact (albeit complicated) non-asymptotic expression was derived.
Moreover, for the specific case of four transmit and two receive
antennas, \cite{pun07b} presented maximum sum-rate scaling laws for the
MRC and MMSE receivers.  In this paper, we generalize these prior
results, by considering the maximum sum-rate for arbitrary antenna
configurations, and deriving new maximum sum-rate results for MRC, ZF and
MMSE receivers. We present new accurate approximations for the
maximum sum-rate, and derive exact maximum sum-rate scaling laws. We show that all three linear receivers converge to the same asymptotic maximum sum-rate as the number of users grows large, however the speed of convergence is different.  Our results are
confirmed through comparison with Monte Carlo simulations.

Throughout this paper, we denote  ${\rm E}[\cdot]$ as expectation, $(\cdot)^\dagger$ as conjugate transpose, ${\rm Tr}(\cdot)$ as matrix trace, $\mathbf{I}_N$ as a $N \times N$ identity matrix and $\mathbf{0}_{N \times M}$ as a $N \times M$ zero matrix. Also, $\otimes$ denotes Kronecker product and $\mathcal{CN}_{M, N}(\mathbf{M},\mathbf{V})$ represents an $M \times N$ matrix-variate complex Gaussian distribution with $M \times N$ mean matrix $\mathbf{M}$ and $MN \times MN$ covariance matrix $\mathbf{V}$.

% we
%the obvious and practical scheduling approach is to
%employ a We next consider a practical scheduling algorithm
%with SINR feedback and low implementation
%complexity. We will show that it is possible for each user to
%calculate its SINR based
%on its own channel measurements, and consider a system framework
%where each user feeds back just their SINR to the transmitter. \color{red} When there is SINR feedback only, an obvious approach would be for the transmitter to employ an opportunistic scheduling approach in which
%each transmit antenna is used to send an independent data stream to
%a corresponding user with the largest SINR. \color{black}

%The
%transmitter employs an opportunistic scheduling approach in which
%each transmit antenna is used to send an independent data stream to
%a corresponding user with the largest SINR.

\section{System Model}

We consider a multiuser MIMO broadcast channel scenario, where a
transmitter has $N_t$ antennas while each of the $K$ users has $N_r$
antennas. Each antenna at the transmitter is used to independently
transmit a stream of data to a particular user, as determined by the two scheduling algorithms described in the next two sections. The
received vector for the $i$th user can be written as
\begin{align}\label{eq:rec_vector}
\mathbf{y}_i &=  \mathbf{H}_i \mathbf{x} + \mathbf{n}_i
\end{align}
where $\mathbf{H}_i \sim \; \mathcal{CN}_{N_r, N_t} \left(
\mathbf{0}_{N_r \times N_t}, \mathbf{I}_{N_r} \otimes
\mathbf{I}_{N_t} \right)$ is the $N_r \times N_t$ Rayleigh fading
channel matrix from the transmitter to the $i$th user, $\mathbf{x}$
is the $N_t \times 1$ transmit symbol vector with ${\rm
E}[\mathbf{x}^\dagger \mathbf{x}]=P$, and $\mathbf{n}_i\sim \mathcal{CN}_{N_r, 1} \left(
\mathbf{0}_{N_r \times 1}, N_0 \mathbf{I}_{N_r} \right)$ is the additive
white Gaussian noise vector. %Where $\mathcal{CN}_{M \times N}(\mathbf{m},\mathbf{V})$ denotes a $M \times N$ complex normal matrix variate with mean $\mathbf{m}$ and covariance $\mathbf{V}$, ${\rm E}[\cdot]$ denotes expectation, $(\cdot)^\dagger$ denotes conjugate transpose, $\mathbf{I}_N$ denotes a $N \times N$ identity matrix, $\mathbf{0}_{N \times 1}$ denotes a $N \times 1$ zero matrix and $\otimes$ denotes Kronecker product. \color{black}
If we
assume that the data stream for the $i$th user comes from the $k$th
transmit antenna, we can write the received vector in
(\ref{eq:rec_vector}) as
\begin{align}
\mathbf{y}_{i,k} = \mathbf{h}_{i,k}x_k + \sum_{j=1, j \neq k}^{N_t}
\mathbf{h}_{i,j} x_j + \mathbf{n}_i
\end{align}
where $\mathbf{h}_{i,j}$ is the $j$th column vector of
$\mathbf{H}_i$ and $x_j$ is the symbol sent from the $j$th transmit
antenna.  We assume that user $i$ perfectly estimates their own
channel matrix, $\mathbf{H}_i$.  We also assume that the total power
budget is distributed equally across the different transmit
antennas, such that ${\rm E}[|x_j|^2]=\frac{P}{N_t}$ for
$j=1,\ldots,N_t$. To recover the desired symbol, we multiply the
received signal by a $1 \times N_r$ received weight vector
$\mathbf{w}^\dagger_i$, which for
%, which results in
%\begin{align}
%\mathbf{r}_i &= \mathbf{w}_i^\dagger \mathbf{y}_i \notag \\
%&= \mathbf{w}_i^\dagger \mathbf{h}_{i,k}x_k + \sum_{j=1, j \neq
%k}^{N_t} \mathbf{w}_i^\dagger \mathbf{h}_{i,j} x_j +
%\mathbf{w}_i^\dagger \mathbf{n}_i \; .
%\end{align}
MRC, ZF, and MMSE linear receivers, and their resulting received SINR, are given respectively by
\begin{align}\label{eq:snr_bf}
& \mathbf{w}_{{\rm MRC}}^\dagger = \mathbf{h}_{i,k}^\dagger
\hspace{1cm} \Rightarrow
& \hspace{1cm}  \gamma_{i, k,{\rm MRC}} =
\frac{\frac{\rho}{N_t} |\mathbf{h}_{i,k}^\dagger
\mathbf{h}_{i,k}|^2}{\frac{\rho}{N_t} \sum_{j=1, j \neq k}^{N_t}
|\mathbf{h}_{i,k}^\dagger \mathbf{h}_{i,j}|^2 +
|\mathbf{h}_{i,k}^\dagger \mathbf{h}_{i,k}|} \; ,
\end{align}
\begin{align}\label{eq:snr_zf}
\mathbf{w}_{{\rm ZF}}^\dagger = \mathbf{g}_k^\dagger \hspace{1cm}
\Rightarrow \hspace{1cm}
\gamma_{i,k,{\rm ZF}} = \frac{ \rho}{N_t \left[\mathbf{H}_i^\dagger \mathbf{H}_i \right]^{-1}_{k,k}}
\end{align}
and
\begin{align}\label{eq:snr_mmse}
& \mathbf{w}_{{\rm MMSE}}^\dagger = \mathbf{h}_{i,k}^\dagger
\left(\mathbf{K} \mathbf{K}^\dagger \frac{\rho}{N_t} +
\mathbf{I}_{N_r} \right)^{-1}
& \hspace{0.2cm} \Rightarrow \hspace{1cm}
\gamma_{i,k,{\rm MMSE}} = \frac{\rho}{N_t} \mathbf{h}_{i,k}^\dagger
\left(\mathbf{K} \mathbf{K}^\dagger \frac{\rho}{N_t} +
\mathbf{I}_{N_r} \right)^{-1} \mathbf{h}_{i,k}
\end{align}
where $\rho = \frac{P}{N_0}$ is the average SNR,
$\mathbf{g}_k^\dagger$ is the $k$th row of $(\mathbf{H}_i^\dagger
\mathbf{H}_i )^{-1} \mathbf{H}_i^\dagger$, $[\mathbf{Z}]_{k,k}$ is
the $(k,k)$th element of $\mathbf{Z}$ and $\mathbf{K} =
\mathbf{H}_i^{\{k\}}$ is the $N_r \times (N_t-1)$ matrix with the
same elements as $\mathbf{H}_i$ but with the $k$th column removed.
Note that throughout this paper, we require $N_r \ge N_t$ when
dealing with ZF receivers, whilst we consider arbitrary antenna
configurations for MRC and MMSE.

It is convenient to define the normalized SINR as follows
\begin{align}
X \defeq \frac{N_t \gamma}{\rho}
\end{align}
where $\gamma$ is the instantaneous SINR given for MRC, ZF and MMSE
in (\ref{eq:snr_bf}), (\ref{eq:snr_zf}) and (\ref{eq:snr_mmse})
respectively.  For the ZF and MMSE receivers, the cumulative
distribution function (c.d.f.) of $X$ for an arbitrary user is given
by \cite{chen07}
\begin{align}\label{eq:cdf_zf}
F_{X_{\rm ZF}}(x) = 1 - e^{-x} \sum_{k=0}^{N_r-N_t} \frac{x^k}{k!}
\end{align}
and \cite{gao98}
\begin{align}\label{eq:cdf_mmse}
F_{X_{\rm MMSE}}(x) =1 - \frac{e^{-x}}{\left(1+\frac{\rho}{N_t}
x\right)^{N_t-1}} \sum_{k=0}^{N_r-1} \beta_k x^k
\end{align}
respectively, where %$\gamma(\cdot,\cdot)$ is the incomplete gamma function,
\begin{align}
\beta_k  = \left( \frac{\rho}{N_t}\right)^k
\sum_{p=\max(0,k-N_t+1)}^k \frac{\binom{N_t-1}{k-p}}{p!
\left(\frac{\rho}{N_t}\right)^p } \; .
\end{align}
For MRC, the c.d.f.\ of $X$ has been previously derived in
\cite{romero08}. However the resulting expression is complicated,
involving multiple summations over sets, and is therefore not easily
amenable to further analysis. In Appendix \ref{app:cdf_bf}, we show
that the SINR c.d.f.\ admits a simpler representation given by
\begin{align}\label{eq:cdf_bf}
F_{X_{\rm MRC}}(x) = 1 -
\frac{e^{-x} }{\left(1 + \frac{\rho}{N_t}x\right)^{N_t-1}} \sum_{k=0}^{N_r-1} \sum_{p=0}^k \frac{\alpha_{p,
k} x^k}{\left(1 + \frac{\rho}{N_t}x\right)^{p}} \;
\end{align}
where
\begin{align}
\alpha_{p,k} =  \frac{\binom{N_t+p-2}{p} \left(\frac{\rho}{N_t}\right)^p}{(k-p)!} \; .
\end{align}
The corresponding probability density functions (p.d.f.) for ZF,
MMSE, and MRC are obtained by taking the derivative of
(\ref{eq:cdf_zf}), (\ref{eq:cdf_mmse}), and (\ref{eq:cdf_bf})
respectively, and are given by
\begin{align}\label{eq:pdf_zf}
f_{X_{\rm ZF}}(x) = \frac{ e^{-x} x^{N_r-N_t}}{(N_r-N_t)!} \; ,
\end{align}
\begin{align}\label{eq:pdf_mmse}
f_{X_{\rm MMSE}}(x) &=\frac{e^{-x}}{\left(1+\frac{\rho}{N_t}
x\right)^{N_t}} \sum_{k=0}^{N_r-1} \beta_k x^{k-1} \left(x\left(1 +
\frac{\rho}{N_t}(N_t+x-1) \right) -k - \frac{\rho k x}{N_t} \right)
\;
\end{align}
and
\begin{align}\label{eq:pdf_bf}
f_{X_{\rm MRC}}(x) &= \frac{e^{-x} }{\left(1 +
\frac{\rho}{N_t}x\right)^{N_t}} \sum_{k=0}^{N_r-1} \sum_{p=0}^k
\frac{\alpha_{p, k} x^{k-1}}{\left(1 + \frac{\rho}{N_t}x\right)^{p}} \left(
x\left(1+\frac{\rho}{N_t}\left(N_t-1+p+x\right)\right)-k-\frac{k x
\rho}{N_t}\right) \; .
\end{align}

In the next two sections, we analyze the maximum sum-rate of the three linear receivers using two different scheduling algorithms which differ depending on whether or not there is SINR feedback from the different users to the transmitter.  We present exact and approximate results, and compare the performance of the three linear receivers. Note that we assume all users employ the same receiver structure (ie. ZF, MMSE, or MRC).

\section{Scheduling Without Feedback}

For multiuser MIMO systems with no feedback capabilities, we
consider a simple scheduling algorithm where each transmit antenna
is assigned to users in a round robin manner. This includes any scenario where each user is served by
each transmit antenna for equal time periods and the data streams
received by the users are decoded independently. We consider the
practical scenario where the number of users is greater than the
number of transmit antennas. The maximum sum-rate under this scheduling
scheme is given by
\begin{align}\label{eq:cap_fair}
R^{\rm rr}(\rho)  = N_t {\rm E}_X\left[\log_2 \left(1 +
\frac{\rho}{N_t} X \right) \right] \;
\end{align}
where the distribution of $X$ is given in (\ref{eq:cdf_zf}), (\ref{eq:cdf_mmse}) and (\ref{eq:cdf_bf}) for the ZF, MMSE and MRC receivers respectively.
Note that this is statistically equivalent to the maximum sum-rate of a single user MIMO spatial
multiplexing system with linear receivers.
Although these systems have been studied extensively, there still remains important open work in this area. Specifically, in this section we present new exact and high SNR closed-form expressions for the maximum sum-rate in (\ref{eq:cap_fair}) with MRC and MMSE receivers. We also present a new analysis of the maximum sum-rate for each of the three linear receivers in the low SNR regime.

\subsection{Exact Maximum Sum-Rate Results}\label{sec:exact_fair_cap}

The exact maximum sum-rate of MIMO systems with ZF receivers has been derived previously (see
eg. \cite{chen07,forenza07}), and is given by
\begin{align}\label{eq:cap_zf_fair}
R_{{\rm ZF}}^{\rm rr} (\rho) =  \log_e 2
\left(\frac{N_t}{\rho}\right)^{N_r-N_t+1} N_t e^{\frac{N_t}{\rho}}
\sum_{n=1}^{N_r-N_t+1} \left(\frac{\rho}{N_t}\right)^n
\Gamma\left(n-N_r+N_t-1,\frac{N_t}{\rho}\right) \;
\end{align}
where $\Gamma(\cdot,\cdot)$ is the incomplete gamma function \cite{abramowitz70}.
For MMSE and MRC receivers however, the corresponding expressions
are not available. To derive these new results, we find it useful to
re-express the maximum sum-rate expression in (\ref{eq:cap_fair}), as given by
the following lemma.
%Here we derive new exact maximum sum-rate results for the MRC and MMSE
%receivers. Note that the exact maximum sum-rate of the ZF receiver has been
%derived previously , however we restate the result here for
%completeness. Our expressions are based on the following lemma which
%gives an alternative form for the .
\begin{lemma}\label{lem:capfair}
The round robin maximum sum-rate can be written as
\begin{align}\label{eq:cap_fair_alt}
R^{\rm rr}(\rho) &= \rho \log_e 2 \int_0^\infty \frac{1 -
F_{X}(x)}{1 + \frac{\rho}{N_t} x} {\rm d} x \; .
\end{align}
\end{lemma}
\begin{proof}
The proof follows by applying integration by parts to
(\ref{eq:cap_fair}).
\end{proof}
In Appendix \ref{app:beamcap}, we employ Lemma \ref{lem:capfair} to
derive the exact maximum sum-rate of MRC receivers as follows
%The maximum sum-rate of the MRC, ZF and MMSE receivers are given
%respectively by
\begin{align}\label{eq:cap_bf_fair}
&R_{{\rm MRC}}^{\rm rr}(\rho) = \log_e 2 \left(\frac{N_t}{\rho}
\right)^{N_t-1}  N_t e^{\frac{N_t}{\rho}} \sum_{k=0}^{N_r-1}
\sum_{p=0}^k \alpha_{p, k} \left(\frac{N_t}{\rho}\right)^k \notag \\
& \hspace{2cm} \times \sum_{q=0}^k \binom{k}{q} (-1)^{k-q}
\left(\frac{N_t}{\rho} \right)^{p-q}
\Gamma\left(1-N_t-p+q,\frac{N_t}{\rho}\right) \; .
\end{align}
Using a similar approach (we omit the specific details), it can be shown
that the exact maximum sum-rate of MMSE receivers is given by
\begin{align}\label{eq:cap_mmse_fair}
&R_{{\rm MMSE}}^{\rm rr}(\rho) = \log_e 2
\left(\frac{N_t}{\rho}\right)^{N_t-1} N_t e^{\frac{N_t}{\rho}}
\sum_{k=0}^{N_r-1} \beta_k \left(\frac{N_t}{\rho}\right)^k \nonumber
\\
& \hspace*{2cm} \times
 \sum_{q=0}^k \binom{k}{q} (-1)^{k-q} \left(\frac{\rho}{N_t}\right)^{q}
 \Gamma\left(1-N_t+q,\frac{N_t}{\rho}\right) \; .
\end{align}
%The proof for the MRC maximum sum-rate is given in Appendix \ref{app:beamcap}
%while the proof for the MMSE maximum sum-rate follows a similar approach to Appendix
%\ref{app:beamcap}.
% verified in test411

Fig.\ \ref{fig:cap_linrec_fair} % and \ref{fig:cap_linrrec_mmsebf}
compares the analytical maximum sum-rate for MRC, MMSE and ZF
based on (\ref{eq:cap_bf_fair}),
(\ref{eq:cap_mmse_fair}) and \cite[Eq. 6]{chen07} respectively. Monte Carlo simulated
curves are also presented for further verification. As expected, we see that the maximum sum-rate of MMSE is greater than for MRC and ZF for all SNR values. We also see an SNR crossover point, such that the maximum sum-rate of MRC is greater than ZF below the crossover point
and less than ZF above the crossover point.  This crossover point occurs because the SNR of MRC and ZF approaches that of MMSE, the optimal linear receiver, at low and high SNR respectively.

We numerically calculate this ZF-MRC SNR crossover point in Fig.\
\ref{fig:crossover}, for different antenna configurations using the
maximum sum-rate expressions for MRC in (\ref{eq:cap_bf_fair}) and ZF in
\cite[Eq. 6]{chen07}. We see that for a fixed $N_r$, increasing
$N_t$ increases the crossover point, however for a fixed $N_t$,
increasing $N_r$ decreases the crossover
point.  We also see that the performance of MRC is better than ZF for a number of practical scenarios. For example, when
$N_r=N_t=4$, the MRC maximum sum-rate is greater than the ZF maximum sum-rate for
SNR values as high as $8$ dB.

We now compare the maximum sum-rate of MRC and MMSE using the analytical expressions in (\ref{eq:cap_bf_fair}) and (\ref{eq:cap_mmse_fair}) respectively. Figs.\ \ref{fig:percent_mmsebf_greater}, \ref{fig:percent_mmsebf_same} and \ref{fig:percent_mmsebf_less} show the percentage of maximum sum-rate achieved by using an MRC receiver with respect to the MMSE receiver (i.e.\ $\frac{R_{{\rm MRC}}^{\rm rr}(\rho)}{R_{{\rm MMSE}}^{\rm rr}(\rho)} \times 100$) for different antenna configurations.  We see in Fig.\ \ref{fig:percent_mmsebf_greater} that when $N_r<N_t$, the maximum sum-rate of MRC is comparable to MMSE, and approaches the MMSE maximum sum-rate for increasing $N_t$. For example, for $N_r=4$ and $N_t=7$, the MRC maximum sum-rate achieves over $80\%$ of the MMSE maximum sum-rate for SNR ranges as high as $5$ dB.  On the other hand, when $N_r \geq N_t$, we notice from Fig.\ \ref{fig:percent_mmsebf_same} and Fig.\ \ref{fig:percent_mmsebf_less} that increasing the number of transmit antennas has the detrimental effect of reducing the percentage of maximum sum-rate achievable. This can be explained by the fact that for the MMSE receiver, when $N_r \ge N_t$, there are enough receive antennas to cancel the interference from the transmit antennas, hence adding more transmit antennas does not decrease the SNR as much as it does for the MRC receiver.

\subsection{High SNR Analysis}

At high SNR, for antenna configurations with $N_r \ge N_t$, expressions for the maximum sum-rate of MMSE receivers have been derived in \cite[Eq. 8.54]{tse05}. For $N_r < N_t$ however, corresponding results are not available.  In this case, the maximum sum-rate of MMSE receivers approach a constant at high SNR. In Appendix \ref{app:bfzfmmse_highsnr}, we derive this constant as follows\footnote{f(x)=o(g(x)) means that $\lim_{x \to \infty} \frac{f(x)}{g(x)} = 0$. }
\begin{align}\label{eq:cap_mmse_highsnr}
R_{{\rm MMSE}}^{\rm rr}(\rho) = \sum_{k=0}^{N_r-1}
\frac{N_t}{N_t-k-1} + o\left(\frac{1}{\rho}\right) \; .
\end{align}
Considering MRC receivers, for $N_t>1$ the maximum sum-rate also approaches a constant at high SNR. Following a similar approach to that used for the MMSE receiver, we evaluate this constant as follows
\begin{align}\label{eq:cap_bf_highsnr}
R_{{\rm MRC}}^{\rm rr}(\rho) = \sum_{k=0}^{N_r-1}
\frac{N_t}{N_t+k-1} + o\left(\frac{\log \rho}{\rho}\right) \; .
\end{align}
For $N_t \le N_r$, the maximum sum-rate of ZF and MMSE scales logarithmically with the SNR \cite[Eq. 8.54]{tse05} instead of approaching a constant, and hence always performs better than MRC as expected. From (\ref{eq:cap_mmse_highsnr}) and (\ref{eq:cap_bf_highsnr}), this can also be shown to be the case when $N_t > N_r$, where the maximum sum-rate of MMSE outperforms MRC.  However, (\ref{eq:cap_mmse_highsnr}) and (\ref{eq:cap_bf_highsnr}) reveal the interesting fact that for high SNR, if $N_r$ is kept fixed and $N_t \to \infty$, the maximum sum-rate
of MRC approaches that of MMSE. This can be intuitively explained by
noting that for MMSE, when $N_t > N_r$, there is not enough degrees
of freedom for the receiver to cancel the interference imposed by
the transmit antennas. Therefore, as $N_t$ grows large, both MMSE
and MRC suffer significant interference penalty. As such, for
sufficiently large $N_t$, the interference cancelation capabilities
of MMSE become insignificant, and the performance of the two
receivers coincide.

The accuracy of (\ref{eq:cap_mmse_highsnr}) and
(\ref{eq:cap_bf_highsnr}) is confirmed in the "Analytical (High SNR)" curves of Fig.
\ref{fig:cap_linrrec_mmsebf}, where we plot these expressions
(without the $o(\cdot)$ terms) as well as Monte Carlo simulated
curves. For all cases considered, we see that the Monte Carlo
simulated curves converge to the analytical asymptotic results in
the high SNR regime.  In addition, we also plot the exact "Analytical" curves from (\ref{eq:cap_bf_fair}) and (\ref{eq:cap_mmse_fair}), and see an exact match with the Monte Carlo simulated curves.

%
%\begin{figure}[htbp]
%\centerline{\includegraphics[width=\columnwidth]{cap_linreceiver_fair_bfmmse.eps}}
%\caption{Maximum sum-rate of MRC, ZF and MMSE using the round robin
%scheduler.  Comparison between analytical and Monte Carlo simulated
%maximum sum-rate for different antenna configurations with $N_t > N_r$.}
%\label{fig:cap_linrrec_mmsebf}
%\end{figure}

\subsection{Low SNR Analysis}\label{subsec:lowsnr}
We now investigate the maximum sum-rate of the three linear receivers in the
low SNR regime. The simplest approach is to derive a first order
Taylor expansion of $R (\rho)$ as $\rho \to 0$. It is shown in
\cite{verdu02} however, that such an approach does not adequately
reveal the impact of the channel, and may in fact lead to misguided
conclusions. Moreover, as discussed in detail in \cite{verdu02}, in
the low SNR (or wideband) regime it is more appropriate to
investigate the maximum sum-rate in terms of the normalized transmit energy
per information bit, $\frac{E_b}{N_0}$, rather than SNR. In this
case the maximum sum-rate representation is given for low $\frac{E_b}{N_0}$
levels by the following expression
\begin{align}
{\sf R}\left(\frac{E_b}{N_0}\right) \defeq S_0 \frac{ \left.
\frac{E_b}{N_0} \right|_{\rm dB} -  \left. {\frac{E_b}{N_0}}_{\min}
\right|_{\rm dB}}{ 3 \; {\rm dB}} + o \left( \frac{E_b}{N_0}
\biggr|_{\rm dB} - {\frac{E_b}{N_0}}_{\min} \biggr|_{\rm dB} \right)
\end{align}
where $\frac{E_b}{N_0}_{\min}$ is the minimum normalized energy per
information bit required to convey any positive rate reliably, and
$S_0$ is the wideband slope in bits/sec/Hz/(3 dB) at the point
$\frac{E_b}{N_0}_{\min}$. Note that $S_0$ and $\frac{E_b}{N_0}_{\min}$ can also be used to determine the bandwidth of a system designed to achieve a certain rate subject to power limitations \cite{verdu02}.

%It can also be shown that the bandwidth
%required to achieve a fixed rate $R$ and power $P$ is well
%approximated by \cite{lozano03}
%\begin{align}\label{eq:bandwidtha}
%B \approx \frac{R}{S_0} \frac{3 {\rm dB}}{ \left. \frac{P}{R N_0}
%\right|_{\rm dB} - \left. \frac{E_b}{N_0}_{\min}  \right|_{\rm dB}}
%\end{align}
%where the approximation sharpens as $\frac{E_b}{N_0} \downarrow
%\frac{E_b}{N_0}_{\min}$.

Now as $\rho \to 0$, it can be shown from (\ref{eq:snr_bf}) and (\ref{eq:snr_mmse}) that the SINR, and hence the maximum sum-rate, of the MMSE
and MRC receivers coincide. In Appendix \ref{app:lowsnr_bf_mmse} we derive closed-form solutions for ${\frac{E_b}{N_{0}}}_{\min}$ and $S_0$ for both MRC and MMSE as \begin{align}\label{eq:ebno_bfmmse}
{\frac{E_b}{N_{0}}}_{\min}^{\rm MRC/MMSE} = \frac{\ln 2 }{N_r}
\end{align}
and
\begin{align}\label{eq:s0_bfmmse}
S_0^{\rm MRC/MMSE} = \frac{2 N_t N_r}{2N_t+N_r-1} \;
\end{align}
respectively.
For ZF receivers, we derive ${\frac{E_b}{N_{0}}}_{\min}$ and $S_0$
as
\begin{align}\label{eq:ebno_zf}
\frac{\ln 2 }{N_r} \; \; \le \; \; {\frac{E_b}{N_{0}}}_{\min}^{\rm
ZF} = \frac{\ln 2}{N_r - N_t + 1} \; \; \le \; \; \ln 2
\end{align}
%\begin{align}\label{eq:ebno_zf}
%{\frac{E_b}{N_{0}}}_{\min}^{\rm MRC/MMSE} \; \; \le \; \;
%{\frac{E_b}{N_{0}}}_{\min}^{\rm ZF} = \frac{\ln 2}{N_r - N_t + 1} \;
%\; \le \; \; \ln 2
%\end{align}
and
\begin{align}\label{eq:s0_zf}
1 \; \; \le \; \; S_0^{\rm ZF} = \frac{2 N_t (N_r-N_t+1)}{N_r-N_t+2} \;
\; \le \; \; \frac{2 N_t N_r}{N_r + 1} \;
%
%1 \; \; \le \; \; S_0^{\rm ZF} = \frac{2 (N_r-N_t+1)}{N_r-N_t+2} \;
%\; \le \; \; S_0^{\rm MRC/MMSE} \;
\end{align}
respectively. The proof is omitted due to space limitations, but involves solving (\ref{eq:ebno_s0}) using the p.d.f.\ in (\ref{eq:pdf_zf}), along with some algebraic manipulation.

We see from (\ref{eq:ebno_bfmmse}) and (\ref{eq:ebno_zf}) that
${\frac{E_b}{N_{0}}}_{\min}^{\rm MRC/MMSE} \le
{\frac{E_b}{N_{0}}}_{\min}^{\rm ZF}$, and from (\ref{eq:s0_bfmmse})
and (\ref{eq:s0_zf}) that $S_0^{\rm MRC/MMSE} \ge S_0^{\rm ZF}$; in
both cases, with equality if $N_t=1$. Therefore, as
expected, whenever $N_t > 1$, the low SNR maximum sum-rate of the MRC and
MMSE receivers always exceeds the maximum sum-rate of ZF. The new results
given by (\ref{eq:ebno_bfmmse})--(\ref{eq:s0_zf}) precisely quantify
this maximum sum-rate difference. In addition, it is interesting to note
that ${\frac{E_b}{N_{0}}}_{\min}^{\rm MRC/MMSE} $ depends only on the number
of receive antennas, also observed in \cite{lozano03} for optimal receivers, whereas ${\frac{E_b}{N_{0}}}_{\min}^{\rm ZF}$ depends on both the number of transmit and receive antennas. In fact, we see that increasing the
number of transmit antennas for a fixed number of receive antennas
has the negative impact of increasing ${\frac{E_b}{N_{0}}}_{\min}^{\rm ZF}$.

%By substituting (\ref{eq:ebno_zf}) and (\ref{eq:s0_zf}) into
%(\ref{eq:bandwidtha}), we can compute the bandwidth $B_{\rm ZF}$
%required to achieve a given rate under a fixed power constraint.  To
%gain further insight, it is instructive to compare the resulting
%expression with the corresponding bandwidth requirements $B_{\rm
%MRC/MMSE}$ for MRC and MMSE, as follows
%\begin{align}\label{eq:bandwidth_compare}
%\frac{2 N_r \left({\left. \frac{P}{R N_0}
%\right|_{\rm dB}}   - \frac{\ln 2}{N_r}\right)}{  (N_r+1) \left({\left. \frac{P}{R N_0}
%\right|_{\rm dB}}   - \ln 2 \right)}  \le  \frac{B_{\rm ZF}}{B_{\rm MRC/MMSE}} = \frac{N_r(N_r-N_t+2) \left({\left. \frac{P}{R N_0}
%\right|_{\rm dB}}   - \frac{\ln 2}{N_r}\right)}{ (N_r-N_t+1) (N_r+1) \left({\left. \frac{P}{R N_0}
%\right|_{\rm dB}}   - \frac{\ln 2}{N_r-N_t+1} \right)}  \le 1 \; .
%\end{align}
%We see in (\ref{eq:bandwidth_compare}) that the upper bound is
%achieved when $N_t=1$ and the lower bound when $N_t=N_r$.
%

%\input{fair_bf_analysis}

\section{Scheduling With SINR Feedback}

For multiuser MIMO systems with feedback capabilities, we consider
a low-rate feedback approach where each user sends to the
transmitter a set of $N_t$ instantaneous SINR values.  These
correspond to the instantaneous SINRs at the output of the linear
receiver; one for each transmit antenna. Specifically, depending on
the particular linear receiver employed, the SINRs are calculated
based on the average SNR and the channel matrix, according to either
(\ref{eq:snr_bf}), (\ref{eq:snr_zf}) or (\ref{eq:snr_mmse}), for
MRC, ZF, and MMSE respectively;
%The SINR for each user is calculated initially by  training
%sequences sent out by each transmit antenna at different times.  For
%each training sequence, the SNR is calculated by each user.
assuming the signal from one antenna is the desired signal while the
signals from the other antennas are interference.

We employ an opportunistic scheduling algorithm which selects,
for each transmit antenna, the user with the maximum SINR. Note that each user can be served by more than one transmit antenna, ie. multiple data streams can be sent to a single user. The maximum sum-rate can be
written as
\begin{align}\label{eq:cap_max}
R^{\max}(\rho) &= \sum_{i=1}^{N_t} {\rm E}_{X_{\max}} \left[
\log_2\left(1
+ \frac{\rho}{N_t} X_{\max}\right)\right] \notag \\
&= N_t {\rm E}_{X_{\max}} \left[ \log_2 \left(1 +
\frac{\rho}{N_t} X_{\max}\right)\right] \notag \\
%&= N_t \int_0^\infty \log(1 + x_{\max}) f_{x_{\max}}(x_{\max}) {\rm d} x_{\max} \notag \\
&= N_t K \int_0^\infty \log_2 (1 + x) f_X(x) F_X^{K-1} (x) {\rm d} x
\end{align}
where $X_{\max} = \max_{1 \le i \le K} X_i$ and $f_X$ and $F_X$ are given for MRC, ZF and MMSE in (\ref{eq:cdf_zf})-(\ref{eq:pdf_bf}).
%In the following sections, we investigate the maximum sum-rate of the
%round-robin scheduler (\ref{eq:cap_fair}) and the maximum SINR
%scheduler (\ref{eq:cap_max}) for MRC, ZF and MMSE receivers.
%In this section, we consider the maximum sum-rate performance when there is
%feedback of the SNR at the transmitter.  Specifically, the
%transmitter sends to the users with the maximum SINR, with the
%maximum sum-rate given in (\ref{eq:cap_max}).
In the following, we derive new accurate approximations for
(\ref{eq:cap_max}) for MRC, ZF and MMSE receivers.  Note that, for
the case of ZF, exact expressions have been derived previously in
\cite{chen07}. Those results, however, incurred high computational
complexity which increased significantly with the number of users
and, as such, would take time comparable to Monte Carlo simulations
to produce meaningful plots.  As such, we are motivated to consider
deriving simpler, more computationally-efficient, accurate
approximations for the ZF receiver as well.

%Hence, in addition to deriving new approximations for
%(\ref{eq:cap_max}) for MRC and MMSE receivers, we are also motivated
%to pursue simple accurate approximations for the ZF receiver as
%well.

%\input{max_low_snr}
%\input{max_high_snr}
%\input{max_exact}
\subsection{Maximum Sum-Rate Approximations}\label{subsec:approx}

To obtain closed form expressions for the maximum sum-rate of the linear
receivers, we can directly evaluate the integral in
(\ref{eq:cap_max}) by first substituting the c.d.f.\ expressions in
(\ref{eq:cdf_zf}), (\ref{eq:cdf_mmse}) and (\ref{eq:cdf_bf}), along
with the corresponding p.d.f.\ expressions
in (\ref{eq:pdf_zf})-(\ref{eq:pdf_bf})
into (\ref{eq:cap_max}). We can then apply the multinomial theorem
and evaluate the resulting integral using standard identities from
\cite{gradshteyn65}. However, although this approach produces an
exact closed-form solution, the resulting expression has a
complexity which increases significantly with the number of users
$K$, and would take time comparable to Monte Carlo simulations to
produce meaningful plots for even small $K$ values (eg. $K = 10$).
As such, we are well motivated to consider accurate approximations
for the maximum sum-rate, which have much lower computational complexity
requirements. Our approach is to approximate (\ref{eq:cap_max}) by
the following
\begin{align}\label{eq:cap_approx}
R^{\max}(\rho)  \approx N_t \log_2 \left(1 + \frac{\rho}{N_t} \, g(K)
\right)
\end{align}
where\footnote{ Note that throughout this paper, we will use $F^{-1}$ to denote the inverse function of $F$, and $F^N$ to indicate the $N$th power. There is no contradiction in this notation since we only consider positive powers in this paper. }
\begin{align}\label{eq:gkdef}
g(K)={F}^{-1}_{X}\left(\frac{1}{1+e^{-\sum_{i=1}^{K-1}\frac{1}{i}}}\right)
\end{align}
with the distribution of $X$ given in (\ref{eq:cdf_zf}), (\ref{eq:cdf_mmse}) and (\ref{eq:cdf_bf}) for the ZF, MMSE and MRC receivers respectively.
This expression is obtained by first noting that for a random
variable $Y$ following a symmetric distribution, with c.d.f.\
$F_Y(\cdot)$ and largest order statistic $Y_K$, the following
inequality holds \cite{vanzwet1}
\begin{align} \label{eq:UpperBound}
\textrm{E}[Y_K] \le F_Y^{-1}
\left(\frac{1}{1+e^{-\sum_{i=1}^{K-1}\frac{1}{i}}}\right) \; .
\end{align}
To obtain (\ref{eq:cap_approx}), we substitute $Y = N_t
\log_2\left(1+ \frac{\rho}{N_t} X\right)$ into
(\ref{eq:UpperBound}), and apply simple algebraic manipulation.

We note that this general upper bounding approach has been shown to
be tight for a Gaussian distribution \cite{order}.  We would
therefore expect it to be tight in (\ref{eq:cap_approx}) also, at
least at low SNR values, upon noting that for low SNR we have
$\log_2\left(1+ \frac{\rho}{N_t} X\right) \approx \frac{\rho}{N_t}
X$, which has been shown to converge to a Gaussian distribution as
$N_t, N_r \to \infty$ for MRC, ZF and MMSE \cite{guo02,zhang01}.
This expected tightness is confirmed in our numerical results.
%In
%fact, our numerical experiments have shown that the bound
%(\ref{eq:cap_approx}) is very tight over the entire range of SNRs.
%[** Can we say this **??]

We will focus on evaluating (\ref{eq:cap_approx}), and show that it
yields an accurate approximation for (\ref{eq:cap_max}). To
calculate (\ref{eq:cap_approx}), we require the inverse function of
$F_{X}(\cdot)$ in $g(K)$. Unfortunately, for ZF, MMSE and MRC, the
c.d.f.s in (\ref{eq:cdf_zf}), (\ref{eq:cdf_mmse}) and (\ref{eq:cdf_bf}) respectively, are in a form for which the
inverse c.d.f.\ is difficult to obtain. We thus employ a simple
numerical algorithm, similar that considered in \cite{louie07}, to
evaluate ${g}(K) ={F}_{X}^{-1}\left( \frac{1}{
1+e^{-\sum_{i=1}^{K-1}\frac{1}{i}} }\right)$.  In particular, we
start with ${g}_0(K) = 1$, and iterate the following expression
\begin{align} \label{eq:iterativeAlg}
{g}_{j+1}(K)= \ln \left(N(K)\right) + \ln \left(1 - F_X\left({g}_{j}(K)\right)\right) + {g}_{j}(K)
\end{align}
until convergence, where $N(K) = 1+e^{\sum_{i=1}^{K-1}\frac{1}{i}}$,
and $F_X(\cdot)$ is given in (\ref{eq:cdf_zf}),
(\ref{eq:cdf_mmse}) and (\ref{eq:cdf_bf}) for ZF, MMSE and MRC receivers respectively. Note that the algorithm will eventually converge, because additional iterations of (\ref{eq:iterativeAlg}) results in additional nested logarithms, which contribute less and less to the outermost logarithm, ie.\ the logarithm at the first iteration at $j=0$.

Simulations indicate that (\ref{eq:iterativeAlg}) has a fast rate of
convergence. The exact rate, however, depends on the number of
antennas and linear receiver used in the system. Table \ref{table1}
shows the absolute error percentage versus the number of iterations,
for the MMSE receiver and for different antenna configurations. The absolute error
percentages were calculated by using the algorithm in
(\ref{eq:iterativeAlg}) and numerical results from the exact c.d.f.\
in (\ref{eq:cdf_mmse}). We see that the convergence rate is fast in
all cases.

We denote $g_{\textrm{conv}}(K)$ to be the converged value from
(\ref{eq:iterativeAlg}).  Substituting this into
(\ref{eq:cap_approx}) in place of $g(K)$ gives an approximation to
maximum sum-rate as follows
\begin{align}\label{eq:cap_approx2}
R^{\max}(\rho) \approx N_t \log_2 \left(1 + \frac{\rho}{N_t} g_{\textrm{conv}}(K)\right) \,
.
\end{align}

Fig. \ref{fig:cap_linrrec_max_mmsebf} compares the analytical
approximations for MRC, ZF and MMSE with Monte Carlo simulated
curves. The analytical approximation curves were obtained using
(\ref{eq:cdf_zf}), (\ref{eq:cdf_mmse}), (\ref{eq:cdf_bf}),
(\ref{eq:iterativeAlg}) and (\ref{eq:cap_approx2}). We see that
these curves accurately approximate the Monte Carlo simulated curves
in all cases.

\subsection{Large User Analysis}

We now examine the maximum sum-rate for MRC, ZF and MMSE, in the
\emph{large-$K$} regime .

\begin{theorem}\label{the:largeuser}
For sufficiently large $K$, the maximum sum-rate for the MRC, ZF, and
MMSE receivers is given by
\begin{align}\label{eq:largeuser}
& R^{\max}(\rho)  = N_t \log_2 \ln K + N_t \log_2
\rho  - N_t \log_2 N_t \notag \\
& \hspace{3cm}  + N_t \log_2 \left(\frac{\ln(\alpha) \rho}{N_t \ln K} +  1 +
\frac{1}{\ln K} \left(1+ \frac{\rho(N_r-N_t)\mathcal{O}(\ln \ln
K)}{N_t} \right)\right)
\end{align}
where $\alpha$ is given for MRC, ZF and MMSE as follows
\begin{equation}\label{eq:alpha_cap}
\alpha =
\begin{cases}
\frac{ \left(\frac{N_t}{\rho} \right)^{N_t-1}}{(N_r-1)!}
\hspace{5.5cm} {\rm MRC} \\
\frac{1}{(N_r-N_t)!} \hspace{5.7cm} {\rm ZF} \\
\sum^{N_r-1} \limits_{p=\max(0,N_r-N_t)} \frac{ \binom{N_t-1}{N_r-1-p}
\left(\frac{\rho}{N_t} \right)^{N_r-N_t-p}}{p!} \hspace{1.2cm} {\rm MMSE} \; .
\end{cases}
\end{equation}
\end{theorem}

\begin{proof}
See Appendix \ref{app:offset}.
\end{proof}

\begin{corollary}\label{corr:largeuser}
For even larger $K$, the maximum sum-rate (\ref{eq:largeuser}) becomes
\begin{align}\label{eq:cap_largeuser}
R^{\max}(\rho) &= N_t \left( \log_2 \ln K + \log_2 \rho - \log_2 N_t
\right)  + \mathcal{O}\left(\frac{\ln \ln K}{\ln K}\right) \; .
\end{align}
\end{corollary}

\begin{proof}
The proof follows by applying basic algebraic manipulation to
(\ref{eq:largeuser}).
\end{proof}

Note that for the case of ZF receivers, an upper bound to
(\ref{eq:cap_largeuser}) was derived in \cite{airy04}, and for the
case of MRC and MMSE receivers with $N_t=4$ and $N_r=2$, the exact
maximum sum-rate scaling law was reported in \cite{pun07b}. As such, our
expressions generalize these prior results by giving exact scaling
laws for all three receivers, which apply for arbitrary numbers of
antennas. Our results also yield important insights into the
fundamental differences in terms of maximum sum-rate scaling for each of the
three linear receivers.

%Our results also identify the key performance parameter
%$\alpha$, in (\ref{eq:alpha_cap}), which characterizes the
%fundamental differences in terms of maximum sum-rate scaling for each of the
%three linear receivers.

From Corollary \ref{corr:largeuser}, we see that the maximum sum-rate
for all three receivers follow the same asymptotic scaling law with
respect to the number of users.  The key difference, however, is
demonstrated by Theorem \ref{the:largeuser}, where the \emph{rate of
maximum sum-rate increase} with $K$ is seen to be a decreasing function of
$\alpha$. That is, although the linear receivers follow the same
asymptotic scaling law, the speed at which that law is approached
varies with $\alpha$, which depends on the particular receiver
structure.

In general terms, we see from (\ref{eq:largeuser}) that for large but finite numbers of users $K$,  the maximum sum-rate is a monotonically increasing function of $\alpha$. From (\ref{eq:alpha_cap}), we clearly see that $\alpha_{\rm MMSE}> \alpha_{\rm MRC}$ and $\alpha_{\rm MMSE} > \alpha_{\rm ZF}$. We also
have $\alpha_{\rm MRC} > \alpha_{\rm ZF}$ if
\begin{align}\label{eq:largeK_condition_bfzf}
\rho \le N_t \left(\frac{(N_r-N_t)!}{(N_r-1)!}
\right)^{\frac{1}{N_t-1}}
\end{align}
and $\alpha_{\rm ZF} >= \alpha_{\rm MRC}$ otherwise. Therefore, for large $K$, (\ref{eq:largeK_condition_bfzf}) represents the SNR crossover point between the maximum sum-rate of ZF and MRC receivers, discussed in Section \ref{sec:exact_fair_cap}.  We note that
the right-hand-side of (\ref{eq:largeK_condition_bfzf}) is a
decreasing function of the number of receive antennas.  This is consistent with the behavior of the  corresponding ZF-MRC SNR crossing point for the round-robin scheduler, observed in Section \ref{sec:exact_fair_cap}.

Note that although this convergence behavior of the three linear
receivers is not readily apparent from the results in Fig.\
\ref{fig:cap_linrrec_max_mmsebf}, for the finite numbers of users
considered, it can be clearly seen in Fig.\
\ref{fig:cap_linrrec_max_diff}, where we plot the difference in
maximum sum-rate between the MMSE and MRC
receivers (i.e.\ $R_{{\rm MMSE}}^{\rm max}(\rho) -R_{{\rm
MRC}}^{\rm max}(\rho)$) as a function of the number of users using
(\ref{eq:cap_approx2}).  We clearly see that the difference in
maximum sum-rate is a decreasing function of the number of users.

The fact that all three linear receivers converge in the limit as $K
\to \infty$ can be intuitively explained as follows.  The
performance in each case is limited by the interference from the
other transmit antennas. As $K \to \infty$, the probability of the
opportunistic scheduling algorithm choosing a user such that the
interference is negligible approaches one. In this interference-free
scenario, all linear receivers thus have the same performance, and
the SINRs all converge to the same value, i.e.\ $\ \gamma_{\rm MRC}
\to \gamma_{\rm ZF} \to \gamma_{\rm MMSE}$ as $K \to \infty$.

We see in Fig. \ref{fig:cap_linrrec_max_mmsebf} that the opportunistic
scheduler can provide a significant increase in maximum sum-rate
compared to a round robin scheduler (i.e.\ corresponding to the point $K=1$); at the expense of
requiring more SINR feedback.  It is also worth noting that when the users are stationary, fairness may become an issue for the SINR-based scheduler, since the transmitter is likely to send to a subset of users for the majority of the time. When the users are
sufficiently mobile, however, such that the channel realizations change for
each transmission period, then on average each user will be served
for approximately equal amounts of time under the opportunistic SINR-based
scheduling approach.

\section{Conclusion}

We have analyzed the maximum sum-rate of MRC, ZF and
MMSE receivers in MIMO multiuser systems. We considered two scheduling algorithms, which differ depending on whether or not there is SINR feedback at the transmitter. When there is no feedback, we considered a simple round-robin scheduling algorithm, deriving new exact closed-form maximum sum-rate expressions for the MRC and MMSE receivers for arbitrary SNRs, and simplified results in the high SNR regime. We also investigated the maximum sum-rate of all three linear receivers in the low SNR regime. When there is feedback, we considered a simple opportunistic scheduling algorithm based on exploiting SINR feedback from the receivers to the transmitter. We derived new accurate maximum sum-rate approximations and exact maximum sum-rate scaling laws for the MRC, ZF and MMSE receivers. Our results demonstrated that all three linear receivers obeyed the same asymptotic scaling law, however the rate of convergence differed based on a simple intuitive parameter $\alpha$, which we quantified. Our results were confirmed through comparison with Monte Carlo simulations.

\begin{appendix}

\subsection{Proof of the SINR C.D.F.\ with
MRC}\label{app:cdf_bf}

The normalized SINR $X_{\rm MRC}$ can be written as
\begin{align}
X_{\rm MRC} = \frac{|\mathbf{h}_{i,k}^\dagger
\mathbf{h}_{i,k}|}{\frac{\rho}{N_t} \sum_{j=1, j \neq k}^{N_t}
\frac{|\mathbf{h}_{i,k}^\dagger
\mathbf{h}_{i,j}|^2}{|\mathbf{h}_{i,k}^\dagger \mathbf{h}_{i,k}|} +
1} = \frac{Z}{Y+1}
\end{align}
where $Z$ is a chi-squared distribution with $2N_r$ degrees of
freedom with c.d.f.\
\begin{align}\label{eq:cdfZ}
F_Z(z) = 1 - e^{-z} \sum_{k=0}^{N_r-1} \frac{z^k}{k!}
\end{align}
and $Y$ is a chi-squared distribution with $2(N_t-1)$ degrees of
freedom with p.d.f.\ \cite{shah98}
\begin{align}\label{eq:pdfY}
f_Y(y) = \frac{y^{N_t-2} e^{-\frac{y
N_t}{\rho}}}{\Gamma\left(N_t-1\right) \left(\frac{\rho}{N_t}
\right)^{N_t-1}} \; .
\end{align}
The c.d.f.\ of $X_{\rm MRC}$ can be written as
\begin{align}\label{eq:cdfproof_bf}
F_{X_{\rm MRC}}(x) &= {\rm Pr} \left( X_{\rm MRC} < x \right) \notag \\
%&= {\rm Pr} \left( Z < x(Y+1) \right) \notag \\
&= \int_0^\infty {\rm Pr} \left( Z < x(y+1) \right) f_Y(y) {\rm d} y \notag \\
%&= 1 - \frac{e^{-x}}{\Gamma(N_t-1) \left(\frac{\rho}{N_t}
%\right)^{N_t-1}} \sum_{k=0}^{N_r-1} \frac{x^k}{k!} \int_0^\infty
%e^{-y\left(x + \frac{N_t}{\rho}\right)} y^{N_t-2} (y+1)^k {\rm d} y
%\notag \\
&= 1 - \frac{e^{-x}}{\Gamma(N_t-1) \left(\frac{\rho}{N_t}
\right)^{N_t-1}} \sum_{k=0}^{N_r-1} \frac{x^k}{k!} \sum_{p=0}^k
\binom{k}{p} \int_0^\infty e^{-y\left(x + \frac{N_t}{\rho}\right)}
y^{N_t-2+p} {\rm d} y \; .
\end{align}
Finally, we obtain the desired result by first solving the integral in (\ref{eq:cdfproof_bf})
using identities in \cite{gradshteyn65} followed by some algebraic manipulation.

\subsection{Proof of the MRC Maximum Sum-Rate with Round Robin Scheduler}\label{app:beamcap}

Substituting (\ref{eq:cdf_bf}) into (\ref{eq:cap_fair_alt}), we can
write the maximum sum-rate of MRC as
\begin{align}
R_{\rm MRC}^{\rm rr}(\rho) &=\sum_{k=0}^{N_r-1} \sum_{p=0}^k \alpha_{p,
k} \int_0^\infty \frac{e^{-x} x^k }{\left(1 + \frac{\rho}{N_t}x\right)^{N_t+p}} {\rm d} x \notag \\
&= e^{\frac{N_t}{\rho}} \sum_{k=0}^{N_r-1} \sum_{p=0}^k \frac{\alpha_{p,
k} N_t^{k+1}}{\rho^{k+1}} \sum_{q=0}^k \binom{k}{q} (-1)^{k-q} \int_1^\infty y^{q-p-N_t} e^{-\frac{y N_t}{\rho}} {\rm d} y
\end{align}
where the first line follows using integration by parts, and
the integral in the last line is solved using integral identities in \cite{gradshteyn65}.

\subsection{Proof of MMSE Maximum Sum-Rate at High SNR for $N_t > N_r$}\label{app:bfzfmmse_highsnr}

We first write the
maximum sum-rate in (\ref{eq:cap_mmse_fair}) as
\begin{align}\label{eq:cap_mmse2}
R_{\rm MMSE}^{\rm rr}(\rho) = N_t \sum_{k=0}^{N_r-1}
\sum_{p=\max(0,k-N_t+1)}^k \frac{\binom{N_t-1}{k-p}}{p! }
\sum_{q=0}^k \binom{k}{q} (-1)^{k-q} \theta(\rho)
\end{align}
where
\begin{align}\label{eq:theta}
\theta(\rho) = e^{\frac{N_t}{\rho}}
\left(\frac{N_t}{\rho}\right)^{\kappa}
\Gamma\left(1-N_t+q,\frac{N_t}{\rho}\right)
\end{align}
and $\kappa = N_t-q-1+p$. To obtain high SNR maximum sum-rate expressions for (\ref{eq:cap_mmse2}), we require the limiting behavior of
$\Gamma\left(1-N_t+q,\frac{N_t}{\rho}\right)$ as $\rho \to \infty$.
This is given by \cite{abramowitz70}
\begin{align}\label{eq:gamma_0}
\lim_{\rho \to \infty} \Gamma\left(1-N_t+q,\frac{N_t}{\rho}\right) &=
\frac{ (-1)^{N_t-1-q}}{ (N_t-1-q)!} \left(\phi(N_t-q) - \log
\left(\frac{N_t}{\rho} \right) \right) \notag \\
& \hspace{2cm} + \frac{
\left(\frac{N_t}{\rho}\right)^{1 - N_t+q}}{N_t-1-q} + o\left(\frac{1}{\rho}\right)\;
\end{align}
where $\phi(\cdot)$ is the digamma function.
Substituting (\ref{eq:gamma_0}) into (\ref{eq:theta}), it is
convenient to write $\theta(\rho)$ as
\begin{align}
\theta(\rho) = \theta_1(\rho) + \theta_2(\rho) +
o\left(\frac{1}{\rho^{\kappa+1}}\right)
\end{align}
where
\begin{align}\label{eq:theta1}
\theta_1(\rho) = \left(\frac{N_t}{\rho}\right)^{\kappa} \frac{
(-1)^{N_t-1-q}}{ (N_t-1-q)!} \left(\phi(N_t-q) - \log
\left(\frac{N_t}{\rho} \right) \right)
\end{align}
and
\begin{align}\label{eq:theta2}
\theta_2(\rho) = \frac{ \left(\frac{N_t}{\rho}\right)^{p}}{N_t-1-q}
\; .
\end{align}
We proceed by first analyzing the high SNR performance of $\theta_1$
and $\theta_2$ in (\ref{eq:theta1}) and (\ref{eq:theta2}).
For $\theta_1$, as $q < N_t-1$ we are only required to find the minimum exponent of $\frac{N_t}{\rho}$ in (\ref{eq:theta1}), i.e.\ minimum $\kappa$. This occurs when $q = N_r-1$, so that $\kappa = N_t-N_r+p$.
Clearly, for any value of $p\ge0$,
\begin{align}\label{eq:theta1_highsnr}
\theta_1(\rho)=o\left(\frac{\log \rho}{\rho}\right) \; .
\end{align}
We now consider $\theta_2(\rho)$. For large
SNR, we only need to consider the smallest exponent of $p$ in (\ref{eq:theta2}), i.e.\ $p=0$.
This gives
\begin{align}\label{eq:theta2_highsnr}
\theta_2(\rho) = \frac{1}{N_t-q-1} +o\left(\frac{1}{\rho}\right)\; .
\end{align}
Substituting (\ref{eq:theta1_highsnr}) and (\ref{eq:theta2_highsnr})
into (\ref{eq:theta}), and then using the resulting expression in
(\ref{eq:cap_mmse2}), we obtain the maximum sum-rate as
\begin{align}
R_{\rm MMSE}^{\rm rr}(\rho) &= N_t \sum_{k=0}^{N_r-1} \binom{N_t-1}{k}
\sum_{q=0}^k \binom{k}{q} \frac{ (-1)^{k-q}}{N_t-q-1} +o\left(\frac{\log \rho}{\rho}\right) =\sum_{k=1}^{N_r} \frac{1}{1-\frac{k}{N_t}} + o\left(\frac{\log
\rho}{\rho}\right)
\end{align}
where the last equality follows after some simple algebraic manipulation.
\subsection{Proof of Low SNR Results}\label{app:lowsnr_bf_mmse}

From \cite{verdu02}, ${\frac{E_b}{N_{0}}}_{\min}$ and $S_0$ are given respectively by
\begin{align}\label{eq:ebno_s0}
{\frac{E_b}{N_{0}}}_{\min} = \frac{\ln 2}{\dot{R}^{\rm rr}(0)} \hspace{1cm} {\rm and} \hspace{1cm} S_0  = - \frac{2 \ln 2 [\dot{R}^{\rm rr}(0)]^2}{\ddot{R}^{\rm rr}(0)}
\end{align}
where $\dot{R}(\cdot)$ and $\ddot{R}(\cdot)$ denote the first and second derivative of $R(\cdot)$ respectively w.r.t.\ $\rho$.

We start with the derivation of ${\frac{E_b}{N_{0}}}_{\min}$ and $S_0$ for the MRC receiver. Substituting (\ref{eq:snr_bf}) into (\ref{eq:cap_fair}) and taking the derivative w.r.t.\ $\rho$, we find that
\begin{align}\label{eq:cap0_rr_mrc_1diff}
\dot{R}^{\rm rr}_{\rm MRC}(0) &= \log_2 (e) {\rm E}\left[|\mathbf{h}_{i,k}^\dagger \mathbf{h}_{i,k} |\right] = N_r \log_2 (e) \; .
\end{align}
%where we have used the distribution of $|\mathbf{h}_{i,k}^\dagger \mathbf{h}_{i,k} |$ given in (\ref{eq:cdfZ}).
Similarly, taking the second derivative of $R^{\rm rr}_{\rm MRC}$ w.r.t.\ $\rho$, we obtain
\begin{align}\label{eq:cap0_rr_mrc_2diff}
\ddot{R}^{\rm rr}_{\rm MRC}(0)& = -\frac{\log_2(e)}{N_t} {\rm E} \left[|\mathbf{h}_{i,k}^\dagger \mathbf{h}_{i,k}| \left(|\mathbf{h}_{i,k}^\dagger \mathbf{h}_{i,k}|+2\sum_{j=1,j\neq k}^{N_t} \frac{|\mathbf{h}_{i,k}^\dagger \mathbf{h}_{i,j}|^2}{|\mathbf{h}_{i,k}^\dagger \mathbf{h}_{i,k}|}\right) \right] \notag \\
&= -\frac{\log_2(e) N_r(2N_t+N_r-1)}{N_t}
\end{align}
where we have used (\ref{eq:cdfZ})
%, as well as the p.d.f.\ of $\sum_{j=1,j\neq k}^{N_t} \frac{|\mathbf{h}_{i,k}^\dagger \mathbf{h}_{i,j}|^2}{|\mathbf{h}_{i,k}^\dagger \mathbf{h}_{i,k}|}$ given in
and (\ref{eq:pdfY}). Substituting (\ref{eq:cap0_rr_mrc_1diff}) and (\ref{eq:cap0_rr_mrc_2diff}) into (\ref{eq:ebno_s0}), we obtain the desired result.

%Now to derive ${\frac{E_b}{N_{0}}}_{\min}$ and $S_0$ for MMSE receivers, we first note that the
Now consider the MMSE receiver.  In this case, it is convenient to write
%
%maximum sum-rate can be written as
\begin{align}\label{eq:cap_mmse_alt}
R^{\rm rr}_{\rm MMSE}(\rho) &= N_t \log_2\left(1- \sum_{i=0}^\infty \left(-\frac{\rho}{N_t}  \right)^{i+1} \mathbf{h}_{i,k}^\dagger \left(\mathbf{K} \mathbf{K}^\dagger \right)^i \mathbf{h}_{i,k} \right)
\end{align}
which is obtained by substituting (\ref{eq:snr_mmse}) into (\ref{eq:cap_fair}), and using \cite{lutkepohl96}. In this form, we can calculate the required first and second order derivatives as follows
 %$\dot{C}(0)$ and $\ddot{C}(0)$ can be easily obtained from $C_{\rm MMSE}^{\rm rr}$, and is given respectively, after some algebraic manipulation as
\begin{align}\label{eq:cap0_rr_mmse_1diff}
\dot{R}^{\rm rr}_{\rm MMSE}(0) &= \log_2(e) {\rm E} \left[\mathbf{h}_{i,k}^\dagger \mathbf{h}_{i,k} \right] = N_r \log_2(e)
\end{align}
and
\begin{align}\label{eq:cap0_rr_mmse_2diffb}
\ddot{R}^{\rm rr}_{\rm MMSE}(0) &= -\frac{ \log_2(e)}{N_t} {\rm E} \left[\left(2  \mathbf{h}_{i,k}^\dagger \mathbf{K} \mathbf{K}^\dagger  \mathbf{h}_{i,k}   +\left( \mathbf{h}_{i,k}^\dagger  \mathbf{h}_{i,k} \right)^2 \right)\right] \; .
\end{align}
To solve (\ref{eq:cap0_rr_mmse_2diff}), we write
\begin{align}\label{eq:mmse_proof1}
{\rm E} \left[\mathbf{h}_{i,k}^\dagger \mathbf{K} \mathbf{K}^\dagger  \mathbf{h}_{i,k} \right] &= {\rm Tr}\left( {\rm E} \left[\mathbf{h}_{i,k}^\dagger \mathbf{K} \mathbf{K}^\dagger  \mathbf{h}_{i,k} \right]\right) \notag \\
&= {\rm Tr}\left( {\rm E} \left[\mathbf{K}^\dagger  \mathbf{h}_{i,k}  \mathbf{h}_{i,k}^\dagger \mathbf{K} \right]\right) \notag \\
%&= {\rm Tr}\left( {\rm E} \left[\mathbf{K}^\dagger \mathbf{K} \right]\right)  {\rm E} \left[ \mathbf{h}_{i,k}^\dagger \mathbf{h}_{i,k}  \right] \notag \\
&= 2(N_t-1)N_r
\end{align}
where we have used \cite[Theorem 2.3.5]{Gupta00}.  Substituting (\ref{eq:mmse_proof1}) into (\ref{eq:cap0_rr_mmse_2diffb}) gives
\begin{align}\label{eq:cap0_rr_mmse_2diff}
\ddot{R}^{\rm rr}_{\rm MMSE}(0) &= -\frac{ \log_2(e)}{N_t} \left(2(N_t-1)N_r +N_r(N_r+1)\right) \notag \\
 &= -\frac{ \log_2(e) N_r \left(2 N_t +N_r-1\right)}{N_t}  \; .
\end{align}
Finally, substituting (\ref{eq:cap0_rr_mmse_1diff}) and (\ref{eq:cap0_rr_mmse_2diff}) into (\ref{eq:ebno_s0}), we obtain the desired result.

\subsection{Proof of Theorem \ref{the:largeuser}}\label{app:offset}

We first give the following lemma:
\begin{lemma}\label{lem:largek}
For sufficiently large $K$,
\begin{align}\label{eq:limitcap}
R^{\max}(\rho) = \log_2\left(1 + \frac{\rho}{N_t}
{F}^{-1}_{X}\left(\frac{K}{K+1} \right)\right)\;.
\end{align}
\end{lemma}

\begin{proof}  The proof follows by applying a general result from
order statistics \cite{vanzwet1} to $C_X^{\max}(\rho)$ in
(\ref{eq:cap_max}), and performing some simple algebraic
manipulation.
\end{proof}

Lemma \ref{lem:largek} implies that for large $K$, we need only consider the c.d.f. in
(\ref{eq:cdf_bf}), (\ref{eq:cdf_zf}) and (\ref{eq:cdf_mmse}) in the
high $x$ regime.   By consider the expressions for the largest exponents of $x$ in (\ref{eq:cdf_bf}), (\ref{eq:cdf_zf}) and (\ref{eq:cdf_mmse}), followed by some algebraic manipulation, this is given for MRC, ZF and MMSE respectively
as
\begin{align}
F_{X_{\rm MRC}}(x) %&= 1 - \frac{ e^{-x} \alpha_{0,N_r-1}
%x^{N_r-1}}{\Gamma(N_t-1) \left(1 + \frac{\rho}{N_t}
%x\right)^{N_t-1}} \notag \\
%&= 1 - \frac{ e^{-x} x^{N_r-1}}{\left(1 + \frac{\rho}{N_t}
%x\right)^{N_t-1} (N_r-1)!} \notag \\
&= 1 - \frac{ e^{-x} x^{N_r-N_t} \left(\frac{N_t}{\rho}
\right)^{N_t-1}}{ (N_r-1)!} + o(x^{N_r-N_t-1})\notag \\
\end{align}
\begin{align}
F_{X_{\rm ZF}}(x) = 1 - \frac{ e^{-x} x^{N_r-N_t}}{(N_r-N_t)!} +
o(x^{N_r-N_t-1})
\end{align}
and
\begin{align}
F_{X_{\rm MMSE}}(x)% &= 1 - \frac{e^{-x} \beta_{N_r-1} x^{N_r-1}}{ \left(
%1+\frac{\rho}{N_t}x\right)^{N_t-1}} \notag \\
&=1 - e^{-x} \beta_{N_r-1} x^{N_r-N_t} \left(
\frac{N_t}{\rho}\right)^{N_t-1}  + o(x^{N_r-N_t-1}) \; .
\end{align}
Note that these c.d.f.\ expansions are each of the general form $F_X^*(x) = 1 - \alpha e^{-x} x^{N_r-N_t} +
o(x^{N_r-N_t-1})$.
In addition, these expansions approach the
actual c.d.f. (without the $o(\cdot)$ term) for large $x$.
Now, the factor ${F}^{-1}_{X}\left(\frac{K}{K+1}\right)$ in
(\ref{eq:limitcap}) can be solved by substituting $N(K) =
K+1$ into the iteration
(\ref{eq:iterativeAlg}). As $K$
is large, only one iteration is needed, since more iterations would
produce nested $\ln$ terms which have negligible effect.
This gives for large $K$
\begin{align}\label{bk2}
R^{\max}(\rho) &= N_t \log_2 \left( 1 + \frac{\rho}{N_t}
\left(\ln(\alpha)
+ \ln(K)  + (N_r-N_t)\mathcal{O}(\ln \ln K) \right)\right) \notag \\
&= N_t \log_2 \ln K + N_t \log_2 \left( \frac{1}{\ln K}  +
\frac{\rho}{N_t} \left(\frac{\ln(\alpha)}{\ln K} + 1 +
\frac{(N_r-N_t)\mathcal{O}(\ln
\ln K)}{\ln K} \right)\right) \; .
%&= N_t \log_2 \ln K + N_t \log_2 \left(\frac{\rho}{N_t} \right) +
%N_t \log_2 \left(\frac{\ln(\alpha)  \rho}{N_t \ln K} +  1 +
%\frac{1}{\ln K} \left(1+ \frac{\rho(N_r-N_t)\mathcal{O}(\ln \ln
%K)}{N_t} \right)\right) \; .
\end{align}
The proof follows by algebraic manipulation.

\end{appendix}

% Generated by IEEEtran.bst, version: 1.12 (2007/01/11)
% Generated by IEEEtran.bst, version: 1.12 (2007/01/11)

\newpage
\begin{figure}[htbp]
\centerline{\includegraphics[width=0.7\columnwidth]{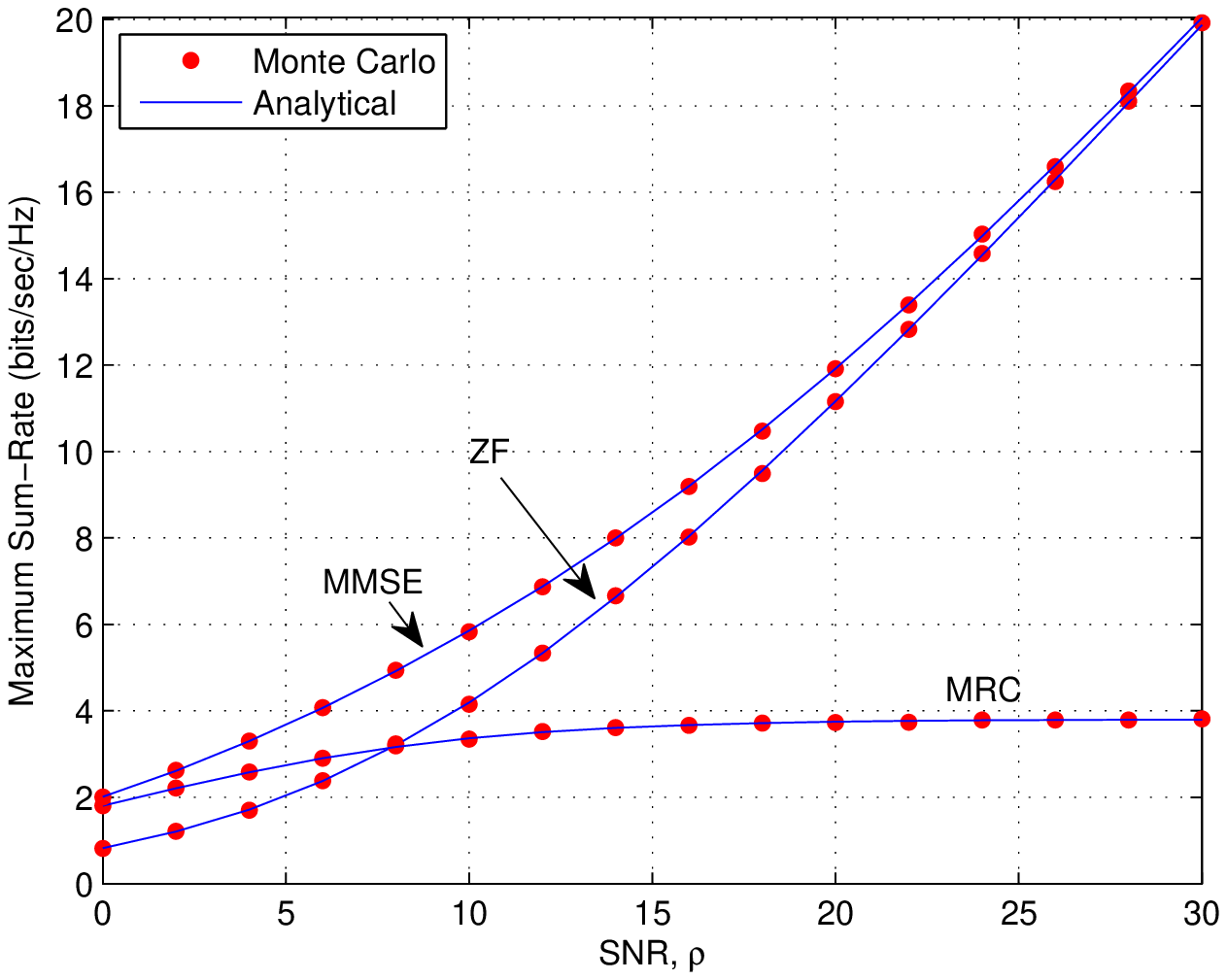}}
\caption{Maximum sum-rate of MRC, ZF and MMSE using scheduling without feedback. Comparison between analytical and Monte Carlo simulated
maximum sum-rate for $N_t=N_r=4$.} \label{fig:cap_linrec_fair}
\end{figure}

\begin{figure}[htbp]
\centerline{\includegraphics[width=0.7\columnwidth]{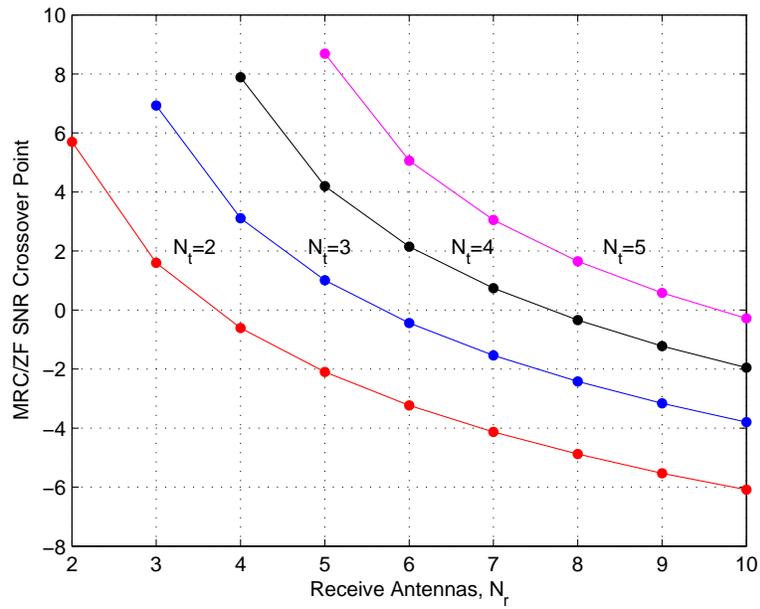}}
\caption{SNR crossover point where the ZF maximum sum-rate is equal to the
MRC maximum sum-rate using scheduling without feedback.} \label{fig:crossover}
\end{figure}

\begin{figure}[htbp]
\centerline{\includegraphics[width=0.7\columnwidth]{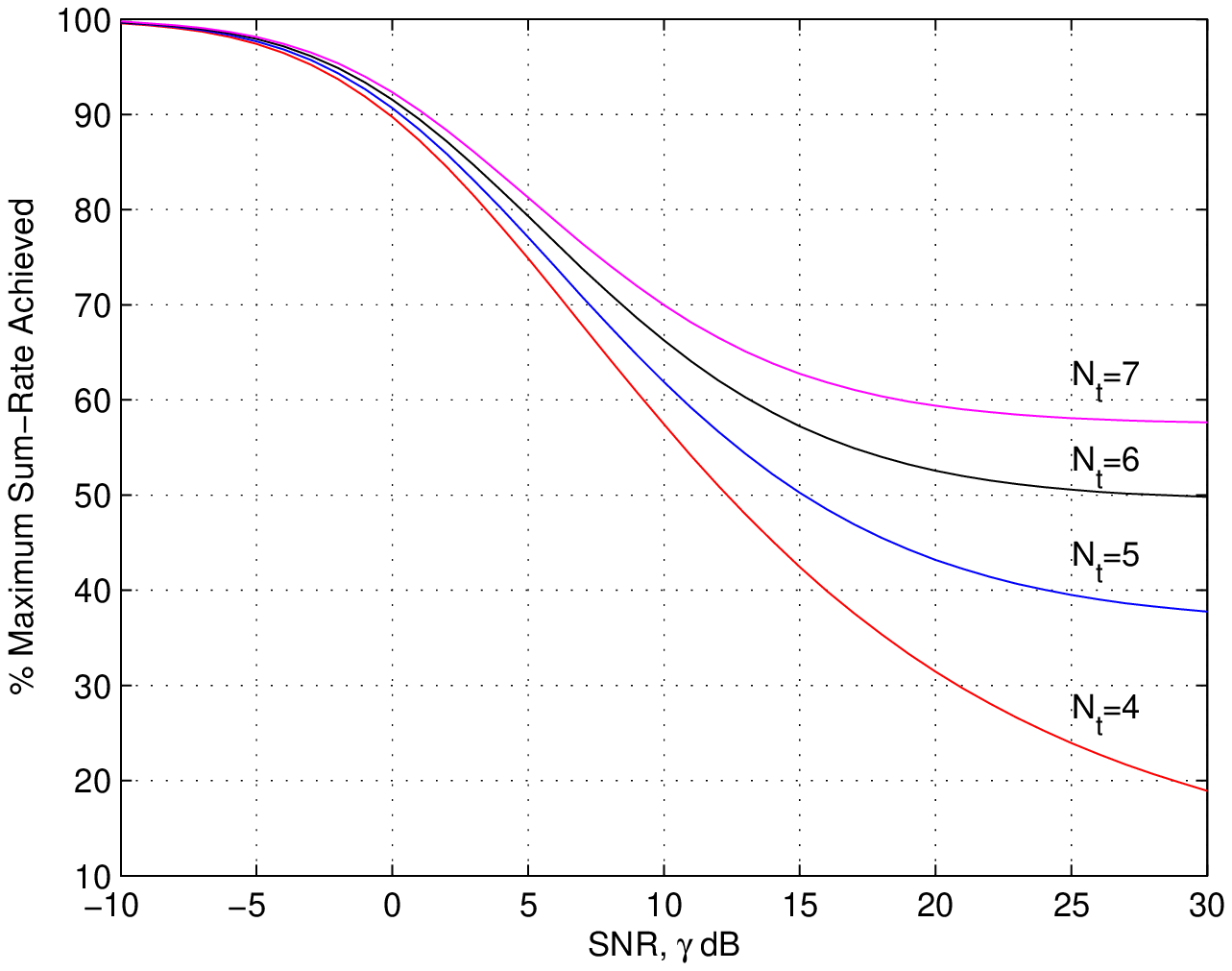}}
\caption{Percentage of MMSE maximum sum-rate achieved by MRC for $N_r=4$ using scheduling without feedback.} \label{fig:percent_mmsebf_greater}
\end{figure}

\begin{figure}[htbp]
\centerline{\includegraphics[width=0.7\columnwidth]{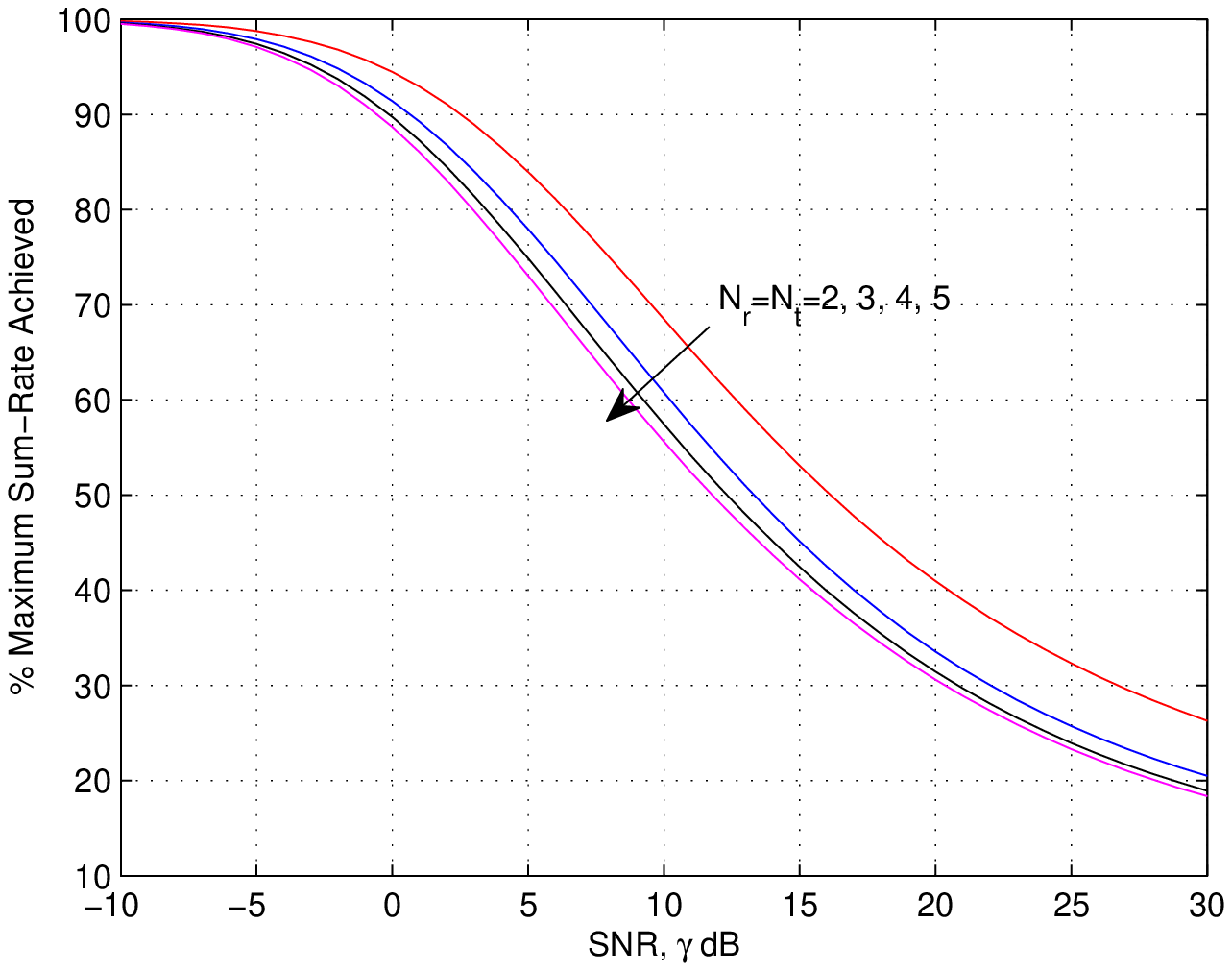}}
\caption{Percentage of MMSE maximum sum-rate achieved by MRC using scheduling without feedback.} \label{fig:percent_mmsebf_same}
\end{figure}

\begin{figure}[htbp]
\centerline{\includegraphics[width=0.7\columnwidth]{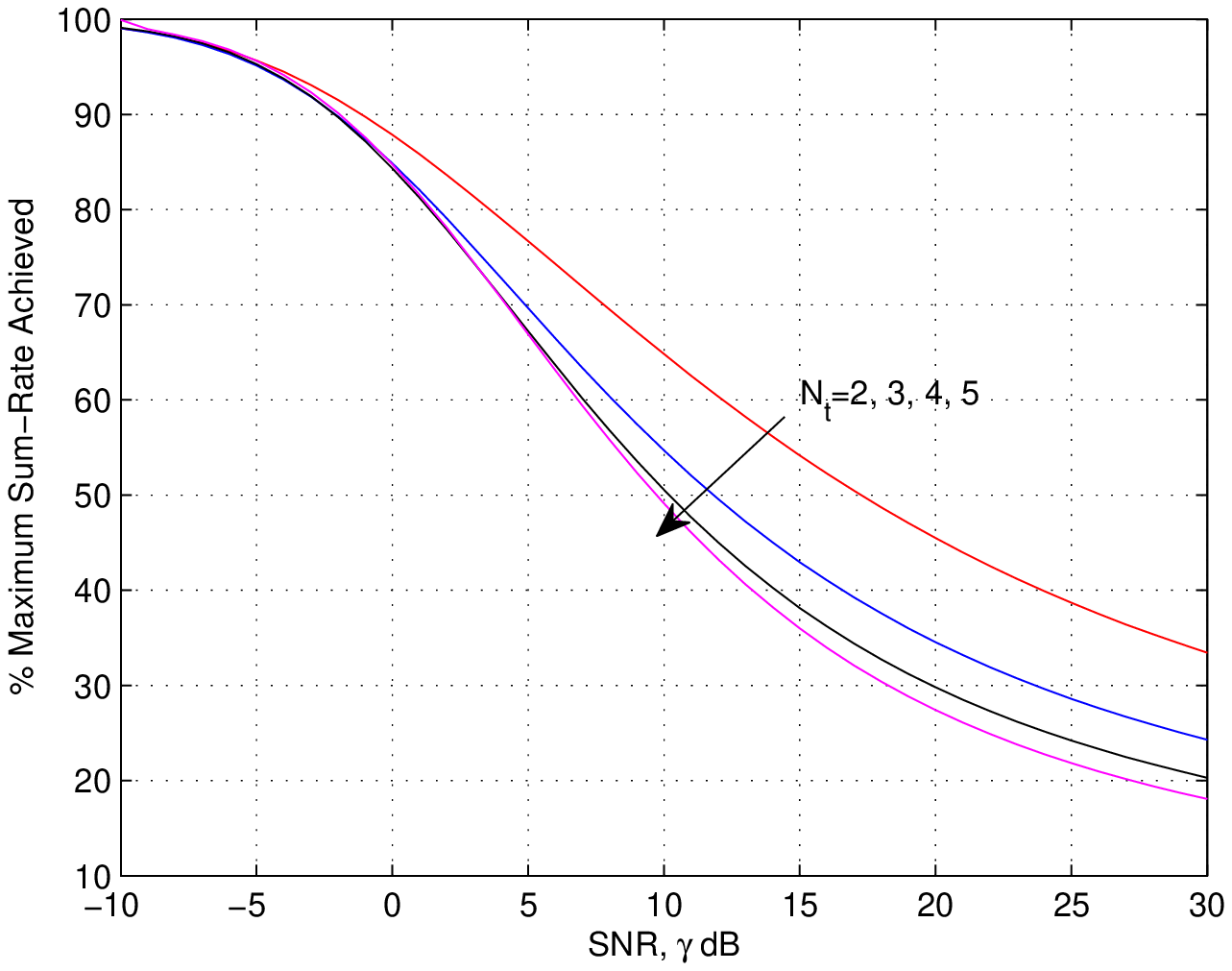}}
\caption{Percentage of MMSE maximum sum-rate achieved by MRC for $N_r=8$ using scheduling without feedback.} \label{fig:percent_mmsebf_less}
\end{figure}

\begin{figure}[htbp]
\centerline{\includegraphics[width=0.7\columnwidth]{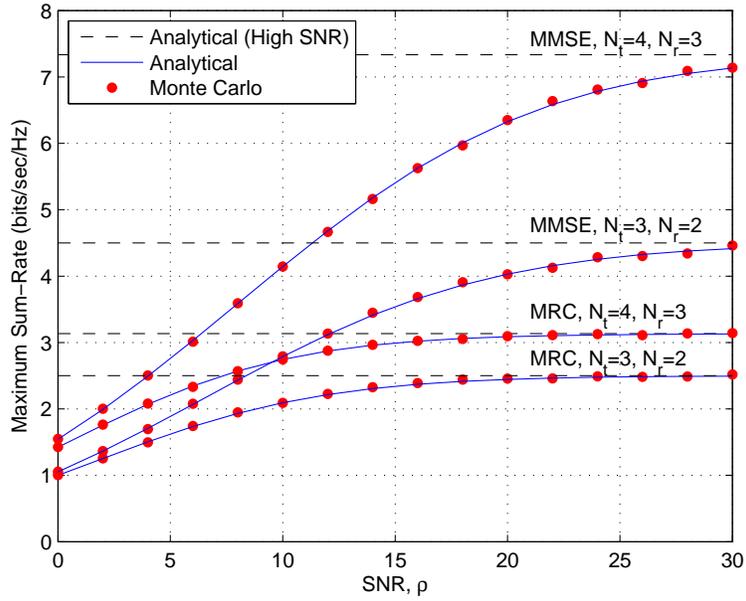}}
\caption{Maximum sum-rate of MRC and MMSE using scheduling without feedback.  Comparison between analytical and Monte Carlo simulated
maximum sum-rate for different antenna configurations with $N_t > N_r$ using scheduling without feedback.}
\label{fig:cap_linrrec_mmsebf}
\end{figure}

%\begin{figure}[htbp]
%\centerline{\includegraphics[width=0.7\columnwidth]{cdf_inverse_iterations.eps}}
%\caption{Absolute error (\%) between exact inverse c.d.f.\ and analytical values using the MMSE receiver, for $K=4$ and $\rho = 3$ dB. } \label{fig:iter_converg}
%\end{figure}

\begin{table} \center
 \caption{Absolute error (\%) between exact inverse c.d.f.\ and analytical values using the MMSE receiver for different iterations, with $K=4$ and $\rho = 3$ ${\rm dB}$. }
   \label{table1}
\begin{tabular}{|c|c|c|c|c|c|c|}
    \hline
Antenna Configuration &  1 It.\ & 2 It.\ & 3 It.\ & 4 It.\ & 5 It.\ & 6 It.\ \\
    \hline
$N_t=N_r=2$ & 10.3845 & 0.8346 & 0.0685 & 0.0056 & 0.0005 & 0.0001    \\
$N_t=N_r=3$ & 21.6231 & 2.7165 & 0.3519 & 0.0458 & 0.0060 & 0.0008    \\
$N_t=N_r=4$ & 32.8893 & 5.1821 & 0.8443 & 0.1383 & 0.0227 & 0.0037    \\
$N_t=N_r=5$ & 43.9760 & 8.0287 & 1.5140 & 0.2877 & 0.0547 & 0.0104    \\
    \hline
\end{tabular}
\end{table}

\begin{figure}[htbp]
\centerline{\includegraphics[width=0.7\columnwidth]{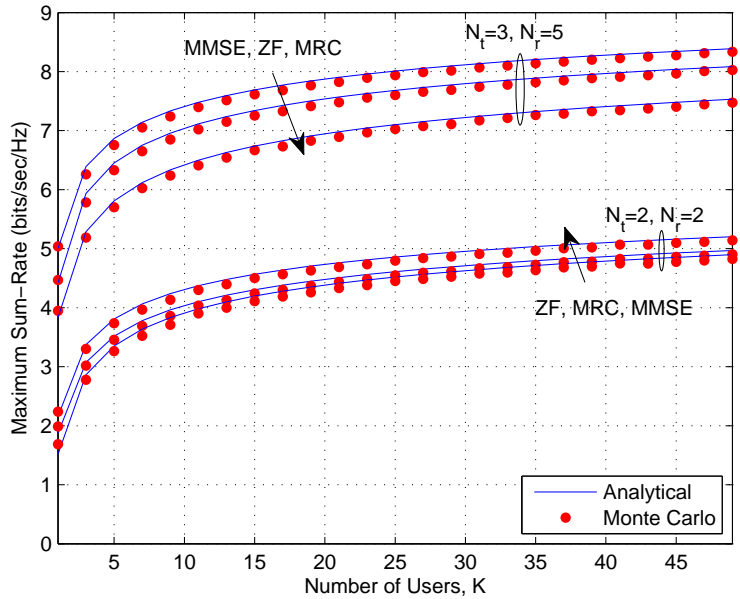}}
\caption{Maximum sum-rate of MRC, ZF and MMSE using the opportunistic scheduler with SINR
feedback. Comparison between analytical and Monte Carlo simulated
maximum sum-rate for different antenna configurations with $\rho=3$ dB.}
\label{fig:cap_linrrec_max_mmsebf}
\end{figure}

\begin{figure}[htbp]
\centerline{\includegraphics[width=0.7\columnwidth]{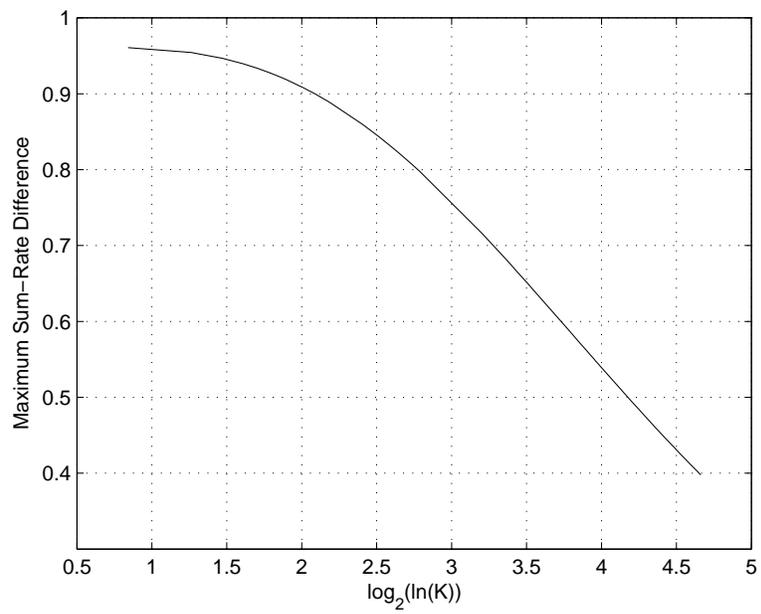}}
\caption{Maximum sum-rate difference between MMSE and MRC vs.\ $\log \ln K$ for $N_t=N_r=4$ and $\rho=3$ dB using the opportunistic scheduler with SINR
feedback.} \label{fig:cap_linrrec_max_diff}
\end{figure}
% plotted in test411_capacity__max_user_linear_receiver_alt_an

\begin{thebibliography}{10}
\providecommand{\url}[1]{#1}
\csname url@samestyle\endcsname
\providecommand{\newblock}{\relax}
\providecommand{\bibinfo}[2]{#2}
\providecommand{\BIBentrySTDinterwordspacing}{\spaceskip=0pt\relax}
\providecommand{\BIBentryALTinterwordstretchfactor}{4}
\providecommand{\BIBentryALTinterwordspacing}{\spaceskip=\fontdimen2\font plus
\BIBentryALTinterwordstretchfactor\fontdimen3\font minus
  \fontdimen4\font\relax}
\providecommand{\BIBforeignlanguage}[2]{{%
\expandafter\ifx\csname l@#1\endcsname\relax
\typeout{** WARNING: IEEEtran.bst: No hyphenation pattern has been}%
\typeout{** loaded for the language `#1'. Using the pattern for}%
\typeout{** the default language instead.}%
\else
\language=\csname l@#1\endcsname
\fi
#2}}
\providecommand{\BIBdecl}{\relax}
\BIBdecl

\bibitem{choi04}
L.-U. Choi and R.~D. Murch, ``A transmit preprocessing technique for multiuser
  {MIMO} systems using a decomposition approach,'' \emph{{IEEE} Trans. Wireless
  Commun.}, vol.~3, no.~1, pp. 20--24, Jan. 2004.

\bibitem{jindal05}
N.~Jindal and A.~Goldsmith, ``Dirty-paper coding versus {TDMA} for {MIMO}
  broadcast channels,'' \emph{{IEEE} Trans. Inform. Theory}, vol.~51, no.~5,
  pp. 1783--1794, May 2005.

\bibitem{wong03}
K.-K. Wong, R.~D. Murch, and K.~B. Letaief, ``A joint-channel diagonalization
  for multiuser {MIMO} antenna systems,'' \emph{{IEEE} Trans. Wireless
  Commun.}, vol.~2, no.~4, pp. 773--786, Jul. 2003.

\bibitem{yoo06}
T.~Yoo and A.~Goldsmith, ``On the optimality of multiantenna broadcast
  scheduling using zero-forcing beamforming,'' \emph{IEEE Sel. Area. Commun.},
  vol.~24, no.~3, pp. 528--541, Mar. 2006.

\bibitem{wang06}
C.~Wang and R.~D. Murch, ``Adaptive downlink multi-user {MIMO} wireless systems
  for correlated channels with imperfect {CSI},'' \emph{{IEEE} Trans. Wireless
  Commun.}, vol.~5, no.~9, pp. 2435--2446, Sep. 2006.

\bibitem{foschini98}
G.~J. Foschini and M.~J. Gans, ``On limits of wireless communications in a
  fading environment when using multiple antennas,'' \emph{Wireless Pers.
  Commun.}, vol. 6, pp. 311--335, Mar. 1998.

\bibitem{mckay_jnl05}
M.~R. McKay and I.~B. Collings, ``General capacity bounds for spatially
  correlated {R}ician {MIMO} channels,'' \emph{{IEEE} Trans. Inform. Theory},
  vol.~51, no.~9, pp. 3121--3145, Sep. 2005.

\bibitem{chen07}
C.-J. Chen and L.-C. Wang, ``Performance analysis of scheduling in multiuser
  {MIMO} systems with zero-forcing receivers,'' \emph{IEEE Sel.\ Area Commun.},
  vol.~25, no.~7, pp. 1435--1445, Sep. 2007.

\bibitem{forenza07}
A.~Forenza, M.~R. McKay, A.~Pandharipande, R.~W. {Heath Jr.}, and I.~B.
  Collings, ``Adaptive {MIMO} transmission for exploiting the capacity of
  spatially correlated channels,'' \emph{IEEE Trans. on Veh. Technol.},
  vol.~56, no.~2, pp. 619--630, Mar. 2007.

\bibitem{tse05}
D.~Tse and P.~Viswanath, \emph{Fundamentals of Wireless Communications},
  1st~ed.\hskip 1em plus 0.5em minus 0.4em\relax New York: Cambridge, 2005.

\bibitem{airy04}
M.~Airy, R.~W. {Heath Jr.}, and S.~Shakkottai, ``Multiuser diversity for the
  multiple antenna broadcast channel with linear receivers: asymptotic
  analysis,'' in \emph{Asilomar Conference on Signals, Systems and Computers},
  vol.~1, Pacific Grove, California, Nov. 2004, pp. 886--890.

\bibitem{pun07b}
M.-O. Pun, V.~Koivunen, and H.~V. Poor, ``{SINR} analysis of opportunistic
  {MIMO-SDMA} downlink systems with linear combining,'' in \emph{IEEE Int.
  Conf. on Commun. (ICC)}, Beijing, China, May 2008, pp. 3720--3724.

\bibitem{gao98}
H.~Gao, P.~J. Smith, and M.~V. Clark, ``Theoretical reliability of {MMSE}
  linear diversity combining in {R}ayleigh-fading additive interference
  channels,'' \emph{{IEEE} Trans. Commun.}, vol.~46, no.~5, pp. 666--672, May
  1998.

\bibitem{romero08}
J.~M. Romero-Jerez and A.~Goldsmith, ``Receive antenna array strategies in
  fading and interference: An outage probability comparison,'' \emph{IEEE
  Trans.\ Wireless Commun.}, vol.~7, no.~3, pp. 920--932, Mar. 2008.

\bibitem{abramowitz70}
M.~Abramowitz and I.~A. Stegun, \emph{Handbook of Mathematical Functions with
  Formulas, Graphs, and Mathematical Tables}, 9th~ed.\hskip 1em plus 0.5em
  minus 0.4em\relax New York: Dover Publications, 1970.

\bibitem{verdu02}
S.~Verd\'{u}, ``Spectral efficiency in the wideband regime,'' \emph{{IEEE}
  Trans. Inform. Theory}, vol.~48, no.~6, pp. 1319--1343, Jun. 2002.

\bibitem{lozano03}
A.~Lozano, A.~M. Tulino, and S.~Verd\'{u}, ``Multiple-antenna capacity in the
  low-power regime,'' \emph{{IEEE} Trans. Inform. Theory}, vol.~49, no.~10, pp.
  2527--2544, Oct. 2003.

\bibitem{gradshteyn65}
I.~S. Gradshteyn and I.~M. Ryzhik, \emph{Table of Integrals, Series, and
  Products}, 4th~ed.\hskip 1em plus 0.5em minus 0.4em\relax San Diego, CA:
  Academic, 1965.

\bibitem{vanzwet1}
W.~V. Zwet, \emph{Convex Transformations of Random Variables}.\hskip 1em plus
  0.5em minus 0.4em\relax Mathematical Centre, 1970.

\bibitem{order}
H.~David and H.~Nagaraja, \emph{Order Statistics}, 3rd~ed.\hskip 1em plus 0.5em
  minus 0.4em\relax New Jersey: John Wiley and Sons, 2003.

\bibitem{guo02}
D.~Guo, S.~Verd\'{u}, and L.~K. Rasmussen, ``Asymptotic normality of linear
  multiuser receiver outputs,'' \emph{{IEEE} Trans. Inform. Theory}, vol.~48,
  no.~12, pp. 3080--3095, Dec. 2002.

\bibitem{zhang01}
J.~Zhang, E.~K.~P. Chong, and D.~N.~C. Tse, ``Output {MAI} distribution of
  linear {MMSE} multiuser receivers in {DS-CDMA} systems,'' \emph{{IEEE} Trans.
  Inform. Theory}, vol.~47, no.~3, pp. 1128--1144, Mar. 2001.

\bibitem{louie07}
R.~H.~Y. Louie, M.~R. McKay, and I.~B. Collings, ``Impact of correlation on the
  capacity of multiple access and broadcast channels with {MIMO-MRC},''
  \emph{{IEEE} Trans. Wireless Commun.}, vol.~7, no.~6, pp. 2397--2407, Jun.
  2008.

\bibitem{shah98}
A.~Shah and A.~M. Haimovich, ``Performance analysis of optimum combining in
  wireless communications with {R}ayleigh fading and cochannel interference,''
  \emph{{IEEE} Trans. Commun.}, vol.~46, no.~4, pp. 473--479, Apr. 1998.

\bibitem{lutkepohl96}
H.~L\'{u}tkepohl, \emph{Handbook of Matrices}, 1st~ed.\hskip 1em plus 0.5em
  minus 0.4em\relax England: John Wiley and Sons, 1996.

\bibitem{Gupta00}
A.~K. Gupta and D.~K. Nagar, \emph{Matrix Variate Distributions}.\hskip 1em
  plus 0.5em minus 0.4em\relax Boca Raton: Chapman \& Hall/CRC, 2000.

\end{thebibliography}
\end{document}